\documentclass[12pt]{amsart}
\usepackage{amsmath, mathabx}
\usepackage{amsxtra}
\usepackage{amscd}
\usepackage{amsthm}
\usepackage{amsfonts}
\usepackage{amssymb}
\usepackage{mathrsfs}
\usepackage{mathbbol}

\usepackage {pstricks}
\usepackage {tikz}
\usepackage{pstricks,pst-node}
\usepackage[all]{xy}
\usepackage{mathbbol}

\usepackage{verbatim}

\usepackage{dashundergaps}
\usepackage{tikz-cd}
\usepackage{enumerate}
\usepackage{graphicx}

\usepackage{epigraph}
\usepackage{physics}

\newtheorem{theorem}{Theorem}[section]

\newtheorem{lemma}[theorem]{Lemma}

\newtheorem{corollary}[theorem]{Corollary}

\theoremstyle{definition}
\newtheorem{definition}[theorem]{Definition}
\newtheorem{example}[theorem]{Example}

\theoremstyle{remark}
\newtheorem{remark}[theorem]{Remark}

\numberwithin{equation}{section}
\numberwithin{figure}{section}

\newcommand{\mc}{\mathcal}

\newcommand{\be}{\begin{equation}}
\newcommand{\ee}{\end{equation}}

\newcommand{\C}{{\mathbb C}}
\newcommand{\R}{{\mathbb R}}
\newcommand{\Z}{{\mathbb Z}}

\newcommand{\K}{{\mathbb K}}

\newcommand{\X}{{\mathbb X}}

\newcommand{\CA}{{\mathcal A}}

\newcommand{\CE}{{\mathcal E}}
\newcommand{\CF}{{\mathcal F}}

\newcommand{\CH}{{\mathcal H}}
\newcommand{\CP}{{\mathcal P}}
\newcommand{\CO}{{\mathcal O}}

\newcommand{\cD}{\mc D}

\newcommand{\SC}{{\mathscr C}}
\newcommand{\SE}{{\mathscr E}}

\newcommand{\SD}{{\mathscr D}}

\newcommand{\mf}{\mathfrak}

\newcommand{\fg}{{\mf g}}
\newcommand{\fh}{{\mf h}}
\newcommand{\fb}{{\mf b}}

\newcommand{\fl}{{\mathfrak l}}

\newcommand{\fp}{{\mathfrak p}}

\newcommand{\fu}{{\mathfrak u}}

\newcommand{\fso}{{\mathfrak {so}}}

\newcommand{\gl}{{\mathfrak {gl}}}
\newcommand{\osp}{{\mathfrak {osp}}}

\newcommand{\Ker}{{\rm{Ker}}}
\newcommand{\id}{{\rm{id}}}

\newcommand{\U}{{\rm{U}}}
\newcommand{\End}{{\rm{End}}}

\newcommand{\GL}{{\rm{GL}}}

\newcommand{\Sym}{{\rm{Sym}}}

\newcommand{\wt}{\widetilde}

\newcommand{\ol}{\overline}

\newcommand{\ot}{\otimes}

\newcommand{\beq}{\begin{eqnarray}}
\newcommand{\eeq}{\end{eqnarray}}

\newcommand{\baln}{\begin{aligned}}
\newcommand{\ealn}{\end{aligned}}

\newcommand{\lra}{\longrightarrow}

\newcommand{\wh}{\widehat}

\newcommand{\FB}{\mathfrak{B}}

\newcommand{\Lam}{\Lambda}

\newcommand{\W}{\mathcal{W}}

\newcommand{\IW}{{{\mathit I}\W}}
\newcommand{\ILam}{{\mathit I}\Lambda}

\newcommand{\CS}{\mathcal{S}}

\newcommand{\BKP}{\K{\mathbb P}}

\begin{document}

\normalfont

\title[Para-spaces]
{Para-spaces, their differential analysis and \\ an application to  Green's quantisation}

\author{R.B. Zhang}
\address{School of Mathematics and Statistics,
The University of Sydney, Sydney, N.S.W. 2006, Australia}
\email{ruibin.zhang@sydney.edu.au}
\date{\today}

\begin{abstract}
We introduce a class of non-commutative geometries, loosely referred to as para-spaces, 
which are manifolds equipped with sheaves 
of non-commutative algebras called para-algebras.
%
%
A differential analysis on para-spaces is investigated, 
which is reminiscent of that on super manifolds  
and can be readily applied to model physical problems, for example, 
by using para-space analogues of differential equations. 
Two families of examples, the affine para-spaces $\K^{m|n}(p)$  
and para-projective spaces $\mathbb{KP}^{m|n}(p)$, 
with $\K$ being $\R$ and $\C$, 
are treated in detail for all positive integers $p$.  
As an application of such non-commutative geometries, we interpret 
Green's theory of parafermions in terms of para-spaces on a point. 
Other potential applications in quantum field theory are also commented upon. 
\end{abstract}
\maketitle


\section{Introduction}\label{sect:intro}
\noindent{1}.
Canonical quantisation imposes commutation relations among  fields 
and conjugate momenta at equal time.  
However, physically meaningful notions such as time evolution and symmetry properties
all involve commutation relations with various currents, 
which are quadratic (in free theories) in fields and conjugate momenta.  
It is therefore more natural to carry out quantisation by imposing equal time commutations relations of currents 
with fields and conjugate momenta, and Green's generalised quantisation \cite{G} does that. 
[Prof. Green explained this to the author in the late 80s.]  
This simple idea had far reaching consequences, 
leading to the celebrated theory of para-statistics \cite{G}, which includes the usual Bose-Fermi statistics 
as the special case  of ``order $1$''. 
There is a huge literature, which is still steadily growing,  
on the subject of para-statistics from the point of view of quantum theory 
(early references can be found in \cite{OK, OK90}). 
There is also a Lie theoretical perspective of para-statistics  \cite{BG, Ca}, 
which too has been widely studied, see, e.g., \cite{BG, Ca, LSV,  P,  SV08} and references therein. 
In recent years there have been reports of quantum simulations of para-particles using theories of 
systems of trapped ions (see \cite{AG} and references therein). This raises the hope that 
some quasi particles appearing in condensed matter may indeed be para-particles. 
 
 %
%
%
%
In quantum mechanics, i.e., field theory in $1$-dimension, fields 
and conjugate moment are coordinates and momenta of particles. 
The usual description of fermions at the classical level takes their  
coordinates (and momenta) as Grassmann variables,
which are coordinate components in odd dimensions of supermanifolds. 
Clearly this is already outside the territory of analysis on usual manifolds.
We mention, in particular, the description of superparticles  (on which there is extensive literature following the original work of Brink and collaborators \cite{BDZ, BdVP}), which is impossible without using supermanifolds. 

\medskip\noindent{2}.
Parafermions include fermions as a special case, thus one is prompted to ask
for the natural ``spaces"  whose coordinates describe parafermions at the classical level,   
and the algebras in which the coordinates of parfermionic are valued. 
%
Satisfactory answers to both questions are necessary in order to 
treat Green's quantisation with mathematical rigour, ideally at a level 
comparable to geometric quantisation or deformation quantisation. 

These questions were investigated
in a series of papers by Ohnuki and Kamefuchi in the 80s starting with \cite{OK80-1, OK80}
(see \cite{OK90} and references therein). 
The authors considered ``para-Grassmann numbers'' by assuming a Green ansatz  \cite[(1.3)]{OK80}
 to map them to an auxiliary ``generalised Grassmann algebra".
 %
Essentials of para-Grassmann numbers are described in Section \ref{sect:graded} (Section \ref{sect:inhom} in particular), 
where one can recognise the auxiliary generalised Grassmann algebra of loc. cit. as 
a $\Z_2^n$-graded commutative algebra for some $p$ and $n$.
%


The Green ansatz inside an auxiliary generalised Grassmann algebra    
is an extremely effective tool for gaining insights into para-Grassmann numbers, 
but the lack of an intrinsic treatment makes para-Grassmann numbers difficult to use if one wants to build 
generalised geometries with them. In \cite{OK80-1, OK80} 
and subsequent papers by the same authors, there were discussions, 
from a physical point of view,  on generalised spaces with 
para-Grassmann numbers as coordinates, which resided inside some 
auxiliary $\Z_2^n$-graded commutative ``space'' related to
the  generalised Grassmann algebra. 

We feel that there is much room to improve our understanding of the above questions, 
especially the geometric aspects of them, which have not been treated to the degree 
of coherence and sophistication required for a mathematical study of Green's quantisation. 

Here we address the questions by proposing a class of intrinsically defined generalised spaces, 
where the nilpotent coordinates have the required properties for describing
parafermions at the classical level. 
The generalised spaces, loosely referred to as para-spaces, 
are manifolds with functions valued in some non-commutative 
algebras called para-algebras introduced in Section \ref{sect:algebra}, that is, they are
ringed spaces with ``structure sheaves'' being sheaves of para-algebras. 
These non-commutative geometries have pertinent differential analysis suitable for physical applications. 
%

%
%
%
\medskip\noindent{3}. 
 Beside offering a natural framework for studying parafermions, para-spaces have other potential applications to quantum theory, which we briefly comment upon here, and hope to investigate in depth in future work.

One can generalise the study of Schr\"odinger equations on superspaces in \cite{D1, Z} to  
para-spaces, as indicated in Example \ref{eg:Schrodinger} and Remark \ref{rmk:Sch}.  
This can potentially produce interesting quantum mechanical models for real world quantum systems. 

Another possible application is to generalise the work \cite{D3, DJW, DZ-2} on
model building in particle physics to address, e.g., the family problem.  
Instead of adding additional Grassmann coordinates as loc. cit.,  
we now use para-algebraic coordinates instead, and possibly also 
include parafermionic generalisations of Poincaré supersymmetry \cite{J}. 

It will be very interesting to extend the Kaluza–Klein theories on $\Z_2$-graded manifolds
\cite{DS, DZ-1, DZ-2} to para-manifolds, thus to approach gravity from a new angle.  
When considering generalised Kaluza–Klein theories in loc. cit., 
one had the lingering question as whether it was possible to build such theories 
by enlarging spacetime with extra coordinates which are nilpotent but their powers 
do not vanish until some given order $p+1\ge 3$. 
The answer to this question is that one should work with para-manifolds. 
It is fairly straightforward to construct such generalised Kaluza–Klein theories 
if we enlarge usual space-time by adding only one para-coordinate.  

\medskip
We now briefly describe the content of this paper. 

\medskip\noindent{4}.
We develop a detailed theory (see Section \ref{sect:algebra}) for a class of non-commutative algebras, 
called para-algebras. For each pair of positive integers $p$ and $n$, we have a para-algebra $\Xi_n(p)$
of order $p$ and degree $n$ (see Definition \ref{def:Xi}). 
It is generated by an $n$ dimensional subspace ${\bf E}$ such that 
for any $\zeta, \zeta'\in {\bf E}$, 
their commutator $[\zeta, \zeta']$  belongs to the centre of $\Xi_n(p)$, 
and $\zeta^p \zeta'= - \zeta' \zeta^p$. 
%
The relationship of para-algebras with para-Grassmann numbers \cite{OK80, OK80-1} 
will be commented upon, see,  in particular, 
Remark \ref{rmk:compare-OK} and
Section \ref{sect:inhom}. 

Despite its non-commutativity, the para-algebra $\Xi_n(p)$ admits 
an analysis reminiscent of Clifford analysis on Grassmann algebras. 
In particular, Theorem \ref{thm:paralg-1} shows the existence of an algebra $\cD_n(p)$, 
referred to as the para Clifford algebra of $\Xi_n(p)$, that is a specific subalgebra 
of the endomorphism algebra of $\Xi_n(p)$, whose elements include
generalised derivations called para-differential operators (see Remark \ref{rmk:para-diff}), 
which have well behaved actions on $\Xi_n(p)$.    
In the generalised geometries be constructed in this paper, para Clifford algebras will play 
a similar role as that of algebras of differential operators in supergeometry. 
We also construct an integral on $\Xi_n(p)$  in Theorem \ref{thm:int-Xi}, 
that reduces to the Berezin integral when $p=1$. 


We think that this theory of para-algebras is an interesting algebraic theory in its own right.

\medskip\noindent{5}.
Given any affine superspace $\K^{m|n}$ of superdimension $m|n$ over a field $\K$ (which is either $\R$ or $\C$), 
we construct, in Section \ref{sect:affine},  a family 
$\K^{m|n}(p)$ (for $p=1, 2, \dots$) of non-commutative generalisations of the superspace, 
called affine para-spaces of order $p$. When $p=1$,  we recover the superspace $\K^{m|n}$. 
We point out that the para-algebra $\Xi_n(p)$ plays the same role in $\K^{m|n}(p)$, as that played by the Grassmann algebra  of degree $n$ (which is $\Z_2$-graded commutative) in the superspace $\K^{m|n}$.

We are able to generalise the differential analysis on the affine superspace $\K^{m|n}$ to the affine para-spaces $\K^{m|n}(p)$ for all $p$. The crucial input is the para-differential operators,  which replace the Grassmann derivations in the odd coordinates of $\K^{m|n}$.   

Particularly noteworthy is that the notion of differential equations generalises to this non-commutative setting, thus one can formulate  and study physical problems using generalised differential equations on para-spaces in the usual way. 
For the purpose of illustration, we solve such a generalised differential equation on $\R^{1|1}(p)$ in Example \ref{eg:Schrodinger}.  

\medskip\noindent{6}.
Following the general philosophy of the algebraic geometric approach to supermanifolds \cite{L}, 
we  glue affine para-spaces together to
construct more complicated non-commutative geometries
in Section \ref{sect:folds}. 
Given any connected manifold $X$, 
we obtain a family $\{\X(p)\mid p=1, 2, ...\}$ of ringed spaces, where the structure sheaf of 
$\X(p)$ is a sheaf of para-algebras of order $p$ over $X$.
We call $\X(p)$ a para-manifold of order $p$, which reduces to a supermanifold when $p=1$. 
Theorem \ref{thm:split} shows that para-manifolds are split in 
the spirit of Bachelor's theorem for $C^\infty$ supermanifolds \cite{B}.  

There is a well behaved local differential analysis for para-manifolds, 
where differential operators and para-differential operators operate on local sections of the structure sheaf. 
We construct in Theorem \ref{thm:transf-oper} well defined transformation rules 
for differential operators and para-differential operators on different coordinate charts,
 thus are able to extend the locally analysis globally.

A family of examples of para-manifolds, called projective para-spaces, are constructed in 
Section \ref{sect:CPmn}, which have the usual projective spaces as underlying manifolds.
We study the projective para-spaces in detail. 

We point out that as para-differential operators are not derivations, the transformation rules given in Theorem \ref{thm:transf-oper}, especially equation \eqref{eq:dx-dy},  are not something which one could easily come up with. 

%
\medskip\noindent{7}.
We  demonstrate  in Section \ref{sect:quantise} the intimate connection between para-spaces and Green's quantisation. 
In Theorem \ref{thm:p-fermi}, we give a formulation of the theory of perafermions \cite{BG, G} 
(in the case with finite degrees of freedom) using para-spaces. 
This shows that para-spaces provide the natural context for perafermions, 
which should help the development of a rigorous mathematical theory for Green's quantisation.  

\medskip\noindent{8}.
Para-manifolds are of intrinsic interest in addition to their relevance to para statistics.
There was enormous effort in investigating ``non-commutative geometries'' in recent decades, 
and much interesting mathematics has been developed, 
i.e., index theorems in the context of Connes theory \cite{C}.  
However, the original goal of finding a non-commutative geometry 
to describe Planck scale physics is still not reachable.

A technical issue is that in most theories of non-commutative geometry, 
non-commutativity invalidates familiar notions of differential analysis on ordinary manifolds 
such as PDEs  (there are exceptions, see, e.g., \cite{CTZZ, GZ, WZZ, ZZ}).  One may approach matters formally, e.g., by following \cite{CD} to introduce generalised notions 
of connections and curvatures on modules over non-commutative algebras, and this can go quite some distance  at a formal level in the context of Connes theory (see \cite{MRS} and reference therein). It still remains to see utilities of such constructs as now PDE type of tools are no longer available.  
The situation is that in the non-commutative world, it is difficult to start from some Riemannian geometry 
to make modifications to achieve a suitable description of space-time at Planck scale.  
As physics is built by gradually modifying existing conceptions of phenomena to 
obtain simpler explanations,  this ostensibly technical problem is in fact quite fundamental. 

Para-manifolds retain almost all features of, and admit differential analysis 
similar to that on, super-manifolds. Thus the above problem does not arise. 

\medskip\noindent{9}.  
This paper was partially inspired by recent developments in two closely related subjects.  
One is $\Z_2^n$-graded super-manifolds, on which there is already quite a number of papers  
(see, e.g., \cite{BIP,  CGP-1}).  
The other is $\Z_2^2$-graded Lie superalgebras (see, e.g., \cite{AD, AIS, SV18, SV23}), which 
are a special class of generalised Lie algebras \cite{RW, Sch, SchZ} amenable to usual Lie theoretical methods. 
Applications of these two subjects to quantum physics are of considerable currently interest 
\cite{AD, AKT1, AKT2, DA}. 

However, the present work 
drives toward a different direction from the above mentioned research on $\Z_2^n$-graded super-manifolds. 
While $\Z_2^n$-graded super-manifolds are $\Z_2^n$-graded 
commutative, para-manifolds are non-commutative, and they are not even $\Z_2^n$-graded for any $n\ge 1$ 
(see Section \ref{sect:differs} for further discussions).

%
%
%
%
%
%


Perhaps classical  field theories  \cite{AKT2, AKT1} defined on $\Z_2^n$-graded super-manifolds \cite{CGP-1} may be relevant through a mechanism involving the Green ansatz as in \cite{OK80-1, OK80} and related papers. 
However, a field theory on $\Z_2^n$-graded super-manifolds 
would need to include fields for Green components. 
For the purpose of treating parafermions, this is not only unwieldy, 
but also has difficulties from the point of view of physics, 
as any field in a model should correspond to matter (that exists in nature). 

The present paper is physically oriented, aiming to convey ideas about para-manifolds 
relevant for physical applications, particularly in Green's quantisation. 
Thus we have kept the technical requirements for the paper at the bare minimum. 

It will be technically much more demanding to develop a general theory of para-manifolds, 
thus we leave the task to a mathematical sequel \cite{Z23-b} of this paper. 
We hope that an analogue of DeWitt's formulation of super manifolds  \cite{dW} 
can be developed for para-manifolds, 
and some constructs in \cite{CD} (and \cite{MRS}) may be adapted to the para-manifold context. 

\medskip
Throughout $\K$ denotes either $\R$ or $\C$, and $\Z_2=\{0, 1\}$.

\section{Para-algebras}\label{sect:algebra}

We introduce a class of non-commutative algebras and study their structure. 
These algebras play an essential role 
in constructing the non-commutative geometries studied in this work.

\subsection{Para-algebras} 

Given any associative algebra, we denote by $[-, -]$  the usual commutator, that is, the bilinear map such that for any elements $A, B$ in the algebra, we have  
$
[A, B]=A B - B A
$.

\begin{definition} \label{def:Xi}
Fix positive integers $n$ and $p$.  Let $\Xi_n(p)$ be a unital associative algebra generated by an $n$-dimensional vector space ${\bf E}$,  subject to the following relations
\beq\label{eq:tensor}
\xi_1^p \xi_2=- \xi_2 \xi_1^p, \quad
[\xi_1, [\xi_2, \xi_3]]=0, \quad \forall\xi_1, \xi_2, \xi_2\in {\bf E}.
\eeq
Call $\Xi_n(p)$ a {\em para-algebra} of order $p$ and degree $n$. 
\end{definition}

The relations \eqref{eq:tensor} are homogeneous, thus the algebra is $\Z_+$-graded, 
$\Xi_n(p)=\sum_{k\ge 0} \Xi_n(p)_k$, with $\Xi_n(p)_0=\K$ and $\Xi_n(p)_1={\bf E}$. 

Note that the first relation in \eqref{eq:tensor} indicates that elements of ${\bf E}$  
behave like ``$p$-th roots" of Grassmann variables. 
This is the defining property of parafermionic fields in conformal field theory.

Let us now consider some special endomorphisms of $\Xi_n(p)$ as a vector space. 
Note in particular that elements of $\Xi_n(p)$ can be regarded 
as endomorphisms acting by left multiplication.  
The following theorem introduces another set of important endomorphisms. 

\begin{theorem}\label{thm:paralg-1} 
Fix a linear isomorphism 
$
\vartheta: \K^n\lra  {\bf E}$, ${\bf a}\mapsto \vartheta_{\bf a}, 
$
and a non-degenerate symmetric
bilinear form $\langle- , - \rangle: \K^n \times \K^n\lra \K$. 
Then there exists a unique subspace ${\bf D}\subset \End_\K(\Xi_n(p))$ isomorphic to $\K^n$ via the linear map $\delta: \K^n\lra {\bf D}$, ${\bf a}\mapsto \delta_{\bf a}$, such that for any ${\bf a}, {\bf b}, {\bf c}\in\K^n$,  
\beq
&&\delta_{\bf a}(1)=0,
\quad \delta_{\bf a}(\vartheta_{\bf b})= p\langle {\bf a}, {\bf b}\rangle, \label{eq:vac-del-1}\\
&&\delta_{\bf a}^p \delta_{\bf b}=- \delta_{\bf b} \delta_{\bf a}^p,   \quad
[\delta_{\bf c}, [\delta_{\bf a}, \delta_{\bf b}]]=0, \label{eq:tensor-del} \\
&&[\delta_{\bf c}, [\vartheta_{\bf a}, \delta_{\bf b}]]=2\langle {\bf c}, {\bf a}\rangle\delta_{\bf b}, 
\quad [\vartheta_{\bf c}, [\vartheta_{\bf a}, \delta_{\bf b}]]=-2\langle {\bf c}, {\bf b}\rangle\vartheta_{\bf a}. \label{eq:tensor-R-2}
\eeq
\end{theorem}

The theorem will be proved in Section \ref{sect:proof-paralg-1}. 
It shows in particular that the formal Definition \ref{def:Xi} for para-algebras as is non-trivial.
Also, it enables us to introduce the following algebra. 
\begin{definition}\label{def:ICliff}
Denote by $\cD_n(p)$ the subalgebra of  $\End_\K(\Xi_n(p))$ generated by ${\bf E}$ 
and ${\bf D}$, and refer to it as a {\em para Clifford algebra}.  
\end{definition}

Now we consider some properties of the para algebra and para Clifford algebra introduced above. 
\begin{lemma}
For any $\xi_1, \xi_2, \dots, \xi_{p+1}\in {\bf E}$, and $\wh{d}_1, \wh{d}_2, \dots, \wh{d}_{p+1}\in {\bf D}$, 
\beq
\sum_{\sigma\in\Sym_{p+1}} \xi_{\sigma(1)} \xi_{\sigma(2)}\dots \xi_{\sigma(p+1)}=0, 
\quad \sum_{\sigma\in\Sym_{p+1}} \wh{d}_{\sigma(1)} \wh{d}_{\sigma(2)}\dots \wh{d}_{\sigma(p+1)}=0, 
\label{eq:tensor-alt}
\eeq
where $\Sym_{p+1}$ is the symmetric group of degree $p+1$. 
\end{lemma}
\begin{proof}
Given arbitrary scalars $x_1, \dots, x_{p+1}$ in $\K$, we let $x^\lambda= \prod_{i=1}^{p+1} x_i^{\lambda_i}$ for any $\lambda=(\lambda_1, \dots, \lambda_{p+1})\in \Z_+^{p+1}$. 

Let $\xi(x_1, \dots, x_{p+1})=\sum_{i=1}^{p+1} x_i \xi_i$. Then  
$\xi(x_1, \dots, x_{p+1})^{p+1}=\sum_\lambda x^\lambda \wh\Theta_\lambda$ 
for some $\wh\Theta_\lambda\in \Xi_n(p)$, where the sum is over compositions of $p+1$, that is, 
those $\lambda$ satisfying $\sum_{i=1}^{p+1} \lambda_i=p+1$. 
It follows from the first relation of \eqref{eq:tensor} that 
$\xi(x_1, \dots, x_{p+1})^{p+1}=0$ for all $(x_1, \dots, x_{p+1})\in\K^{p+1}$. 
Since $x^\lambda$ for distinct partitions $\lambda$ of $p+1$ are linearly independent as functions on $\K^n$, 
it follows that $\wh\Theta_\lambda=0$ for every composition $\lambda$ of $p+1$. 
Now $\wh\Theta_{(1, 1, \dots, 1)}=0$ leads to the first relation in \eqref{eq:tensor-alt}. 

The second relation in \eqref{eq:tensor-alt} can be proved in the same way. 
\end{proof}

\begin{remark} \label{rmk:compare-OK}
Para-Grassmann numbers were in principle defined by \cite[(1.1)]{OK80}, 
even though they were always realised by using the Green ansatz \cite[(1.3)]{OK80} in practice. 
The relations in \cite[(1.1)]{OK80} coincide with the first relation of \eqref{eq:tensor-alt} 
and second relation of \eqref{eq:tensor} in the present context. 
With the Green ansatz realisation \cite[(1.3)]{OK80}, para-Grassmann numbers also 
obey the first relation of  \eqref{eq:tensor} by Lemma \ref{lem:ILam}. 
\end{remark}

Take the standard basis $\{e_i\mid i=1, 2, \dots, n\}$ for $\K^n$.
Then the elements  $\vartheta_i=\vartheta_{e_i}$ for $i=1, 2, \dots, n$ form a basis of ${\bf E}$, 
thus generate $\Xi_n(p)$. By \eqref{eq:tensor}, they satisfy the following relations.
\beq
&& \vartheta_i^p \vartheta_j= - \vartheta_j \vartheta_i^p,   \label{eq:the-p-0}\\
&& [\vartheta_k, [\vartheta_i, \vartheta_j]]=0,   \quad \forall i, j, k. \label{eq:R-0}
\eeq

For any non-negative integer $\ell$  and $j_1, j_2, \dots, j_\ell\in \{1, 2, \dots, n\}$,
we let  
\beq\label{eq:products}
\wh\varTheta_{j_1, j_2, \dots, j_\ell}=\vartheta_{j_1} \vartheta_{j_2}\dots \vartheta_{j_\ell}. 
\eeq
All such elements span $\Xi_n(p)$.  
Denote 
$
d_j( j_1, j_2, \dots, j_\ell)=\sum_{r=1}^\ell \delta_{j j_r}. 
$

\begin{theorem} \label{thm:vanish}
Retain notation above. Then 
\[
\wh\varTheta_{j_1, j_2, \dots, j_\ell}=0 \  \text{ if $d_j( j_1, j_2, \dots, j_\ell)\ge p+1$ for any $j$}.
\] 
This in particular implies $\dim_\K(\Xi_n(p))<\infty$. 
\end{theorem}
\begin{proof}
It follows from the $i=j$ case of  \eqref{eq:the-p-0} that 
\beq\label{eq:vanish-p-1}
\vartheta_i^{p+1}=0, \quad \forall i. 
\eeq
 Now by  using \eqref{eq:R-0}, we obtain
$[\vartheta_j, \vartheta_i^r] = r  \vartheta_i^{r-1} [\vartheta_j, \vartheta_i]$ for all $r$.
Induction on $r$ shows that for all $s\le r$, and $j_1, j_2, \dots, j_s\in \{1, 2, \dots, n\}$,
\beq
[\vartheta_{j_s}, \dots [\vartheta_{j_2}, [\vartheta_{j_1}, \vartheta_i^r ]] \dots] 
= \frac{r!}{(r-s)!} \vartheta_i^{r-s} \prod_{t=1}^s [\vartheta_{j_t}, \vartheta_i].
\eeq
Setting $r=p+1$ in the above relation, and then using \eqref{eq:vanish-p-1}, we obtain 
\beq\label{eq:vanish-p-s}
\vartheta_i^{p+1-s} \prod_{t=1}^s [\vartheta_{j_t}, \vartheta_i]=0.
\eeq

Consider $\wh\varTheta_{j_1, j_2, \dots, j_\ell}$. Assume that there is some $i$ such that $d_i=d_i( j_1, j_2, \dots, j_\ell)\ge p+1$. We move all $\vartheta_i$ factors in $\wh\varTheta_{j_1, j_2, \dots, j_\ell}$ to the far left. Then $\wh\varTheta_{j_1, j_2, \dots, j_\ell}$ becomes a linear combination of elements of the form 
\[
\vartheta_i^{d_i-s} \prod_{t=1}^s [\vartheta_{j'_t}, \vartheta_i] \wh\varTheta_{j'_{s+1}, j'_{s+2}, \dots, j'_{\ell-d_i}}, 
\]
where the indices $j'_t$'s are the $j_r$'s which are not equal to $i$. Such elements are $0$ by 
\eqref{eq:vanish-p-s}.  This proves the first statement of the theorem, 
and the second follows immediately. 
\end{proof}  

The dimension of $\Xi_n(p)$ is known   \cite{BG}, which is given by \eqref{eq:fd}.

To consider properties of  the algebra $\cD_n(p)$, 
it is most convenient to take the bilinear form such that 
the standard basis is orthonormal relative to it, whence $\langle e_i , e_j \rangle=\delta_{i j}$ for all $i, j$. 
Now the
elements $\delta_i:=\delta_{e_i}\in \End_\K(\Xi_n(p))$, for $1 \le i \le n$, form a basis of ${\bf D}$, and  
satisfy the relations
\beq
&& \delta_i(1)=0, \quad   \delta_i(\vartheta_j)=p\delta_{i j}, \quad \forall i, j,  \label{eq:vac-1}, \\
&&\delta_i^p \delta_j = - \delta_j \delta_i^p,   \label{eq:del-p-0}  \\
&& [\delta_k, [\delta_i, \delta_j]]=0,  \label{eq:R-1}\\
&&[\delta_k, [\vartheta_i, \delta_j]]=2\delta_{k i}\delta_j, 
\quad [\vartheta_k, [\vartheta_i, \delta_j]]= -2\delta_{ k j} \vartheta_i. \label{eq:R-2}
\eeq
Furthermore,  \eqref{eq:R-2}, \eqref{eq:R-0} and \eqref{eq:R-1}  imply the following relations. 
\beq
&&   [\delta_k, [\vartheta_i, \vartheta_j]]=2\delta_{k i}\vartheta_j-2\delta_{k j} \vartheta_i, \label{eq:R-3}\\
&& [\vartheta_k, [\delta_i, \delta_j]] = 2\delta_{k i} \delta_j - 2\delta_{k j} \delta_i. \label{eq:R-4}
\eeq

The para Clifford algebra has a $\Z$-grading.  Let
$
\hat{d}=\frac{1}2\sum_{i=1}^n \left([\vartheta_i, \delta_i]-p\right)
$. 
We have 
\beq
[\hat{d}, \vartheta]=\vartheta, \quad [\hat{d}, \delta]=-\delta, \quad  \forall \vartheta\in{\bf E}, \ \delta\in{\bf D}. 
\eeq
Thus the operator provides a $\Z$-grading for $\cD_n(p)$ with
\beq
&&\cD_n(p)=\sum_{-np\le \ell \le  np} \cD_n(p)_\ell, \quad \cD_n(p)_\ell=\{A\in \cD_n(p)\mid [ \hat{d}, A]=\ell A \}.
\eeq

The operator $\wh{d}$ also leads to the $\Z_+$-grading for $\Xi_n(p)$ discussed earlier with 
$ \Xi_n(p)_\ell=\{F\in \Xi_n(p)\mid [ \hat{d}, F]=\ell F \}$. 
Note that $\Xi_n(p)_\ell$ is spanned by  all $\wh\varTheta_{j_1, j_2, \dots, j_\ell}$ for
$j_1, j_2, \dots, j_\ell\in \{1, 2, \dots, n\}$. 

\medskip
\noindent{\bf A consistency check}. 
Relations \eqref{eq:vac-1} and the second relation of \eqref{eq:R-2} will lead to 
$\delta_i(\vartheta_i^{p+1}) = c \vartheta_i^p$ for some scalar $c$, and consistency with 
\eqref{eq:vanish-p-1}  requires $c=0$. 
%
%
This is a rather non-trivial matter. In fact 
such consistency  in the $p=1$ case underpinned the very definition of Grassmann derivations.  
To verify the consistency,  note that 
\[
\baln
\delta_i(\vartheta_i^{p+1}) &= \sum_{k=0}^p \vartheta_i^{p-k}[\delta_i, \vartheta_i](\vartheta_i^k).
\ealn
\]
Equations \eqref{eq:vac-1} and \eqref{eq:R-2} imply
$
[\delta_i, \vartheta_i](\vartheta_i^k)=(p-2k) \vartheta_i^k.
$
Hence 
\[
\delta_i(\vartheta_i^{p+1}) =  \sum_{k=0}^p  (p-2k) \vartheta_i^p = 0, 
\]
proving consistency.

\subsection{Para Clifford algebras}

We now prove Theorem \ref{thm:paralg-1} and investigate further properties of $\cD_n(p)$.
We do this by finding a non-zero algebra homomorphism $\cD_n(p)\lra \W_n(P)$, 
where $\W_n(p)$ is the $\Z_2^p$-graded Clifford algebra of Definition \ref{def:Wn}. 
The proof is inspired by the Green ansatz in the theory of para-statistics \cite{G}, 
and makes use of 
the class of $\Z_2^p$-graded commutative algebras discussed in Appendix \ref{sect:graded}, 
where one can find all unexplained notation which appears in this subsection.

\subsubsection{Existence of para Clifford algebras}\label{sect:para-Clifford}

Consider the subalgebra  $\ILam_n(p)$ of
$\Lam_n(p)$ defined in Section \ref{sect:inhom}. 

\begin{lemma}\label{lem:Xi-ILam}
There exists an algebra epimorphism $\iota_0: \Xi_n(p)\lra \ILam_n(p)$ such that
$\iota_0(\vartheta_i)=\theta_i$ for all $i$. 
\end{lemma}
\begin{proof}
It immediately follows from Lemma \ref{lem:ILam} that the map $\iota_0$ preserves the defining relations of $\Xi_n(p)$, hence follows the lemma.
\end{proof}

\begin{remark}
Note that the map $\iota_0$ is a $\Z_+$-graded algebra epimorphism, where the $\Z_+$-grading of $\Xi_n(p)$ is described in Remark \ref{rmk:Z-grade}. 
\end{remark}

Now we consider the subalgebra  $\IW_n(p)$  of  $\W_n(p)$ defined in Section \ref{sect:inhom}. 
The following result immediately follows from Lemma \ref{lem:Xi-ILam}, Lemma \ref{lem:IW} and equation \eqref{eq:Del-p-1}. 
\begin{theorem} \label{thm:epi-1} Retain notation above. 
There exists an algebra epimorphism $\iota: \cD_n(p)\lra \IW_n(p)$ such that 
$
\vartheta_i\mapsto \theta_i$ and $\delta_i\mapsto \Delta_i$ for all $i.
$
Its restriction to $\ILam_n(p)$ coincides with the algebra epimorphism $\iota_0:\Xi_n(p)\lra\ILam_n(p)$  in Lemm \ref{lem:Xi-ILam}. 
\end{theorem}

This in particular implies the existence of the para Clifford algebra $\cD_n(p)$ for all $p, n$. 

Now we can complete the proof of Theorem \ref{thm:paralg-1}.

\subsubsection{Proof of Theorem \ref{thm:paralg-1}}\label{sect:proof-paralg-1}
\begin{proof} 
The existence of $\cD_n(p)$, and hence of 
 ${\bf D}\subset \End_\K(\Xi_n(p))$, was established in Theorem \ref{thm:epi-1}. 
To complete the proof of Theorem \ref{thm:paralg-1}, we show that ${\bf D}$ is unique by proving the uniqueness of $\delta_i$. 
This is done by showing that 
their actions on $\Xi_n(p)$ are uniquely defined by the given relations in Theorem \ref{thm:paralg-1}.

We use induction on the degree $\ell$ with respect to the $\Z_+$-grading of $\Xi_n(p)$ 
to show the uniqueness of the action of  $\delta_i$ on $\Xi_n(p)_\ell$. 
The cases $\ell=0$ and $\ell=1$ are evident. 
For any $\ell>1$,  the elements 
$\wh\varTheta_{j_1, j_2, \dots, j_\ell}$, with $j_1, j_2, \dots, j_\ell\in \{1, 2, \dots, n\}$, 
span $\Xi_n(p)_\ell$. 
We have 
\[
\baln
\delta_i\left(\wh\varTheta_{j_1, j_2, \dots, j_\ell}\right)
&= \sum_{r=1}^{\ell}  \wh\varTheta_{j_1, \dots, j_{r-1}}[\delta_i, \vartheta_{j_r}] \left(\wh\varTheta_{j_{r+1}, \dots, j_\ell}\right), 
\ealn
\]
where $[\delta_i, \vartheta_{j_r}]\in \cD_n(p)$,  thus acts on $\wh\varTheta_{j_{r+1}, \dots, j_\ell}$. We re-write the right hand side as 
\[
\baln
\sum_{r=1}^{\ell} \sum_{s={r+1}}^\ell  \wh\varTheta_{j_1, \dots, \wh{j_r}, \dots, j_{s-1} }
[[\delta_i, \vartheta_{j_r}],  \vartheta_{j_s}] 
\left(\wh\varTheta_{j_{s+1}, \dots, j_\ell}\right)
+ p  \sum_{r=1}^{\ell} \delta_{i j_r}   \wh\varTheta_{j_1, \dots, \wh{j_r}, \dots,  j_\ell}, 
\ealn
\]
where the notation $\wh{j_r}$ indicates that the entry $j_r$ is deleted from the sequences where it appeared. 
Using the second relation of \eqref{eq:R-3}, we obtain
\[
\baln
\delta_i\left(\wh\varTheta_{j_1, j_2, \dots, j_\ell}\right)&=
 -2 \sum_{r=1}^{\ell} \sum_{s={r+1}}^\ell \delta_{i j_s}  \wh\varTheta_{j_1, \dots, \wh{j_r}, \dots, j_{s-1},  j_r,  j_{s+1}, \dots, j_\ell}\\
&+ p  \sum_{r=1}^{\ell} \delta_{i j_r}  \wh\varTheta_{j_1, \dots, \wh{j_r}, \dots,  j_\ell}.  
\ealn
\]
This shows that the action of $\delta_i$ on $\Xi_n(p)_\ell$ is uniquely defined,  
completing the proof of Theorem \ref{thm:paralg-1}.
\end{proof}

\subsubsection{A realisation of para Clifford algebras}\label{sect:realise}

Let us consdier the natural action of $\IW_n(p)$ on $\Lam_n(p)$. 
\begin{theorem}\label{thm:mod-1}
The natural $\IW_n(p)$-action on $\Lambda_n(p)$ restricts to an action on $\ILam_n(p)$, which is simple. 
\end{theorem}
\begin{proof}
Consider the operator 
$
d=\frac{1}2\sum_{i=1}^n \left([\theta_i, \Delta_i]-p\right), 
$
which satisfies $[d, \theta_i^{(\alpha)}]=\theta_i^{(\alpha)}$ and $\left[d, \frac{\partial}{\partial\theta_i^{(\alpha)}}\right]=- \frac{\partial}{\partial\theta_i^{(\alpha)}}$ for all $i$ and $\alpha$.  Thus it gives rise to a 
$\Z$-grading $\W_n(p)=\sum_{k=-np}^{np} \W_n(p)_k$ for $\W_n(p)$, where $\W_n(p)_k=\{F\in \W_n(p)\mid d(F)=h F\}$. Restricting to $\Lam_n(p)$, we obtain a $\Z_+$-grading $\Lam_n(p)=\sum_{d=0}^{n p} \Lam_n(p)_d$.

Note that $[d, \theta_i]=\theta_i$ and $\left[d, \Delta_i\right]=- \Delta_i$ for all $i$. 
Hence the above gradings for $\W_n(p)$ and $\Lam_n(p)$ descend to the subalgebras 
$\IW_n(p)\subset \W_n(p)$ and $\ILam_n(p)\subset \Lam_n(p)$ respectively. 
In particular, $\ILam_n(p)=\sum_{d=0}^{n p} \ILam_n(p)_d$ with $\ILam_n(p)_d=\ILam_n(p)\cap\Lam_n(p)_d$. 

Applying $\Delta_i$ to $1$ and $\theta_j$, we obtain 
\eqref{eq:1}, where the first relation immediately follows from the definition of $\Delta_i$, and  the second is easily verified by using \eqref{eq:Del-the}.

 Now we show that $\ILam_n(p)$ is stable under the actions of all $\Delta_i$. This implies that $\ILam_n(p)$ forms an $\IW_n(p)$-module. 

For any fixed $\ell$, we have the following spanning set of $\ILam_n(p)_\ell$
\[
\Theta_{j_1, j_2, \dots, j_\ell}:=\iota(  \wh\varTheta_{j_1, j_2, \dots, j_\ell}) 
= \theta_{j_1} \theta_{j_2}\dots \theta_{j_\ell}, \quad \forall j_1, j_2, \dots, j_\ell\in \{1, 2, \dots, n\}.
\]
The computations in the proof of Theorem \ref{thm:epi-1} imply
\[
\baln
\Delta_i\left(\Theta_{j_1, j_2, \dots, j_\ell}\right)= -2 \sum_{r=1}^{\ell} \sum_{s={r+1}}^\ell \delta_{i j_s}  \Theta_{j_1, \dots, \wh{j_r}, \dots, j_{s-1},  j_r,  j_{s+1}, \dots, j_\ell}
+ p  \sum_{r=1}^{\ell} \delta_{i j_r}   \Theta_{j_1, \dots, \wh{j_r}, \dots,  j_\ell}, 
\ealn
\]
where both terms on the right hand side belong to $\ILam_n(p)_{\ell-1}$. 
Therefore $\ILam_n(p)$ is stable under the actions of all $\Delta_i$,  thus forms a $\IW_n(p)$-module.

Note that the claimed simplicity of $\ILam_n(p)$ as $\IW_n(p)$-module is well-known for $p=1$, 
but requires proof when $p>1$.  Consider the $\Z_+$-grading of $\ILam_n(p)$. 
The homogeneous component $\ILam_n(p)_{n p}$ of the highest degree $n p$ is $1$-dimensional, 
which is spanned by $\theta_1^p \theta_2^p\dots \theta_n^p=(p!)^n\Theta^{\SE_{max}}$ 
(see Remark \ref{rmk:vanish} and \eqref{eq:Theta-E}  for notation). 
It is evident that every non-zero submodule contains $\ILam_n(p)_{n p}$.
The lowest degree homogeneous component is $\ILam_n(p)_0=\C$, which generates
$\ILam_n(p)$ as $\IW_n(p)$-module. We prove simplicity of $\ILam_n(p)$ by 
showing that all non-zero submodules contain $\ILam_n(p)_0$.

We claim that the operator $D$ defined by \eqref{eq:diff-int} can be re-written as 
\beq\label{eq:diff-int-1}
D=\frac{1}{(p!)^n}\Delta_n^p \Delta_{n-1}^p \dots \Delta_1^p, 
\eeq
and hence $D\in \IW_n(p)$.  
This formula can be easily verified by first using equation \eqref{eq:Del-p-1} to the right hand side, and then re-arrange the orders of the derivations by using the fact that if $\alpha\ne\beta$, then $\frac{\partial}{\partial\theta_i^{(\alpha)}} \frac{\partial}{\partial\theta_j^{(\beta)}}=\frac{\partial}{\partial\theta_j^{(\beta)}}  \frac{\partial}{\partial\theta_i^{(\alpha)}}$ for all $i, j$. Now the obvious fact $D(\Theta^{\SE_{max}})=1$ leads to 
\[
\frac{1}{(p!)^n}\Delta_n^p \Delta_{n-1}^p \dots \Delta_1^p(\theta_1^p \theta_2^p\dots \theta_n^p) = (p!)^n, 
\]
proving that every non-zero submodule contains $\ILam_n(p)_0$.  

This completes the proof. 
\end{proof}

The following result follows from Theorem \ref{thm:mod-1}.
\begin{corollary} Retain notation above. 
\begin{enumerate}[i).]
\item As associative algebra, $\IW_n(p)=\End_\K(\ILam_n(p))$. 
\item
The map  $\iota$ induces an action of $\cD_n(p)$ on $\ILam_n(p)$, which is irreducible. 
\end{enumerate}
\begin{proof}  The first part is clear from Theorem \ref{thm:mod-1}. 
 As the algebra homomorphism $\iota: \cD_n(p)\lra \ILam_n(p)$ is surjective by Theorem \ref{thm:epi-1}, 
 it follows from Theorem \ref{thm:mod-1} that  $\iota$ induces an irreducible $\cD_n(p)$-action on $\Lam_n(p)$.
 This proves the corollary. 
\end{proof}
\end{corollary}

By definition, $\cD_n(p)$ is a subalgebra of $\End_\K(\Xi_n(p))$. 
Thus it naturally acts on $\Xi_n(p)$.    
We have the following result. 
\begin{theorem}\label{thm:paralg-2} 
As an $\cD_n(p)$-module with the natural action, $\Xi_n(p)$ is simple and is finite dimensional. Hence 
$\cD_n(p)=\End_\K(\Xi_n(p))$.
\end{theorem}

This theorem easily implies the following result. It gives a realisation of $\cD_n(p)$ in terms of $\Z_2^p$-graded algebras and their graded derivations. 
\begin{theorem}\label{thm:iso-2} 
The algebra homomorphism $\iota: \cD_n(p)\lra\IW_n(p)$ of Theorem \ref{thm:epi-1} is an isomorphism. 
\end{theorem}
\begin{proof}
We have the following commutative diagram, 
\beq
\begin{tikzcd}
\cD_n(p)\ot \Xi_n(p) \arrow[d, "\iota\ot\iota_0"' pos=0.43] \arrow[r]& \Xi_n(p)\arrow[d, "\iota_0"]\\
\IW_n(p)\ot \ILam_n(p) \arrow[r]&  \ILam_n(p), 
\end{tikzcd}
\eeq
where the horizontal maps are the natural actions. The right vertical map  is a $\cD_n(p)$-module epimorphism, where the $\cD_n(p)$-action on $\ILam_n(p)$ is that induced by $\iota$. 
Since $\Xi_n(p)$ is a simple $\cD_n(p)$-module by Theorem \ref{thm:paralg-2}, we must have $\ker(\iota_0)=0$, and $\iota_0$ is an isomorphism. Hence $\End_\K(\Xi_n(p))\cong \End_\K(\ILam_n(p))=\IW_n(p)$. 
The claimed isomorphism between $\cD_n(p)$ and $\IW_n(p)$ now follows from the second statement of Theorem \ref{thm:paralg-2}. 
\end{proof}

The following fact is implied by Theorem \ref{thm:iso-2}. It was also shown independently in the proof of Theorem \ref{thm:iso-2}. 
\begin{lemma}\label{cor:iso-3}
The algebra homomorphism $\iota_0: \Xi_n(p)\lra \ILam_n(p)$ of Lemma \ref{lem:Xi-ILam} is an isomorphism. 
\end{lemma}

We now turn to the proof of Theorem \ref{thm:paralg-2}. 

\subsection{Proof of Theorem \ref{thm:paralg-2}}\label{sect:proof-paralg-2}

The proof is  Lie theoretical  in nature and requires some preparation. 

\subsubsection{A Lie theoretical perspective}

Any associative algebra can be given a Lie algebra structure with the usual commutator $[-  , - ]$ as the Lie bracket. In particular, we give the para Clifford algebra $\cD_n(p)$ such a Lie algebra structure. 
Denote by $\fg$ the subspace of $\cD_n(p)$ with the basis $B=B_+\cup B_0\cup B_-$, where 
\[
\baln
B_+&= \{\vartheta_i\mid 1\le i\le n\}\cup \{[\vartheta_i, \vartheta_j]\mid 1\le i< j\le n \}, \\
B_0&=\{[\vartheta_i, \delta_j]\mid i, j=1, 2, \dots, n\}, \\
B_-&=\{\delta_i \mid 1\le i\le n\}\cup \{ [\delta_i, \delta_j] \mid 1\le i< j\le n \}.
\ealn
\]
The relations \eqref{eq:the-p-0}, \eqref{eq:R-0}, \eqref{eq:del-p-0}, \eqref{eq:R-1}, \eqref{eq:R-2}, \eqref{eq:R-3} and \eqref{eq:R-4} in particular show that $\fg$ is closed under the commutator, thus  is a Lie subalgebra.  

Let $\fu=\K B_+$, $\ol\fu=\K B_-$ and $\fl=\K B_0$, then we have the decomposition $\fg=\fu+\fl+\ol\fu$. Note that $\fp=\fu+\fl$ is a parabolic subalgebra and $\fl$ is its Levi subalgebra. Write $E_{i j}=\frac{1}2[\vartheta_i, \delta_j]$ for all $i, j$. Using Lemma \ref{lem:IW}, we can show that these elements satisfy the standard commutation relations of $\gl_n(\K)$, i.e., 
\beq\label{eq:gl}
[E_{i j}, E_{k \ell}]= \delta_{j k} E_{i \ell} -  \delta_{i \ell} E_{k j},    
\eeq 
and hence $\fl\cong \gl_n(\K)$. 

By considering the the root datum of $\fg$ with respect to the Boral and Cartan subalgebras $\fb\supset \fh$, where $\fb=\fb(\fl)+\fu$ with $\fb(\fl)=\sum_{i\le j}\K E_{i j}\supset \fh=\sum\K E_{i i}$,  we can see that $\fg\cong \fso_{2n+1}(\K)$. 
Thus we have proved the following result
\begin{lemma}\label{lem:Lie}  Retain notation above. 
The subspace $\fg\subset \cD_n(p)$ equipped with the usual commutator $[ -  , -  ]$ forms a Lie algebra isomorphic to $\fso_{2n+1}(\C)$. 
\end{lemma}

We can consider $\Xi_n(p)$ as a $\fg$-module. We have the following result. 

\begin{lemma}  \label{lem:g-mod}
As $\fg$-module, $\Xi_n(p)$ is isomorphic to the simple module with integral dominant highest weight  $\lambda=(\frac{p}2, \frac{p}2, \dots, \frac{p}2)$. 
\begin{proof}
As $\fg$-module, $\Xi_n(p)$ is cyclically generated by $1$, and has the following properties. 
\[
\baln
&\delta_i(1)=0, \ \text{ and hence } \  [\delta_i, \delta_j](1)=0, \quad \forall i\ne j, \\
& E_{i j}(1) = \frac{1}{2}[\vartheta_i, \delta_j] (1)=0, \quad \forall i\ne j,  \\
& E_{i i}(1)=\frac{1}2[\vartheta_i, \delta_i](1) = - \frac{p}{2}, \quad \forall i. 
\ealn
\]
In particular,  any root vector in the Borel subalgebra $\ol\fb$ (opposite Borel of $\fb$)  annihilates $1$, 
where $\ol\fb$ is given by 
$
\ol\fb=\ol\fb(\fl)+\ol\fu$, with  $\ol\fb(\fl)=\sum_{i\ge j}\K E_{i j}.
$
Hence $\Xi_n(p)$ is a quotient module of a Verma module $\ol{V}_{\ol\lambda}$ induced by a $1$-dimensional $\ol\fb$-module spanned by a vector with weight $\ol\lambda=(-\frac{p}2, -\frac{p}2, \dots, -\frac{p}2)$.

Note that the weight $\ol\lambda$ is integral anti-dominant. Thus the quotient of $\ol{V}_{\ol\lambda}$ by the unique maximal submodule $\ol{M}_{\ol\lambda}$ is finite dimensional. Since it is the only finite dimensional quotient module of $\ol{V}_{\ol\lambda}$, we have $\Xi_n(p)\cong \ol{V}_{\ol\lambda}/\ol{M}_{\ol\lambda}$, which is the simple $\U(\fg)$-module with lowest weight $\ol\lambda$,  and highest weight $\lambda=(\frac{p}2, \frac{p}2, \dots, \frac{p}2)$. 
\end{proof}
\end{lemma}

It follows Lemma \ref{lem:g-mod} that the dimension of $\Xi_n(p)$ is given by Weyl's dimension formula, 
and we have
\beq\label{eq:fd}
 \dim_\K(\Xi_n(p))=\prod_{1\le i\le j\le n}\frac{p+i+j-1}{i+j-1}.
\eeq

\subsubsection{$\Xi_n(p)$  as $GL_n(\K)$-module}\label{sect:GL-mod}
The decomposition of $\Xi_n(p)$ as $\gl_n(\K)$-module is known \cite{BG}  (also see \cite{GP}). 
To discuss it, we work with the $\gl_n(\K)$ algebra in $\cD_n(p)$ 
with the generators 
\[ 
\wt{E}_{ij} = \frac{1}2[\vartheta_i, \delta_j] + \frac{p}2\delta_{i j}, \quad i, j=1, 2, \dots, n,  
\]
which also satisfy the standard commutation relations of $\gl_n(\K)$. 
Then the weights of $\Xi_n(p)$ for this $\gl_n(\K)$ all belong to $\{0, 1, \dots, p\}^n$. 
In particular, $\wt{E}_{ii}(1)=0$ for all $i$. Denote by $L_\lambda$ a simple $\gl_n(\K)$-submodule of $\Xi_n(p)$ with highest weight $\lambda=(\lambda_1, \lambda_2, \dots, \lambda_n)$. 
Then as $\gl_n(\K)$-module,  
\beq\label{eq:gl-mod}
\Xi_n(p) = \bigoplus_{\lambda} L_\lambda, 
\eeq
where the sum is over all $\lambda\in\CP(n; p)$ with 
\beq
\CP(n; p)=\left\{\lambda\in \{0, 1, \dots, p\}^n\mid p\ge \lambda_1\ge \lambda_2\ge \dots\ge \lambda_n\ge 0\right\}, 
\eeq
and every $L_\lambda$ with such a highest weight $\lambda$ appears with multiplicity $1$. 
Write $|\lambda|=\sum_{i=1}^n \lambda_i$. Then 
\[
\Xi_n(p)_\ell = \bigoplus_{\substack{\lambda\in\CP(n; p), \\ |\lambda|=\ell}} L_\lambda
\]
Denote by $\rho^{(n; p)}_\lambda$ the representation of $\GL_n(\K)$ on $L_\lambda$, 
 and write 
\beq\label{eq:GL-rep}
\rho^{(n; p)}_{(\ell)}=\bigoplus_{\substack{\lambda\in\CP(n; p), \\ |\lambda|=\ell}} \rho^{(n; p)}_\lambda, \quad \rho^{(n; p)}=\bigoplus_{\ell} \rho^{(n; p)}_{(\ell)}.
\eeq

All $\rho^{(n; p)}_\lambda$ are polynomial representations of $\GL_n(\K)$, and hence $\rho^{(n; p)}$ is also a polynomial representation.

\begin{remark}\label{rmk:GL-rep}
Given any commutative $\K$-algebra $A$ with identity, there exists a polynomial representation of $\GL_n(A)$ on $A\ot_K L_\lambda$, which we still denote by $\rho^{(n; p)}_\lambda$. Hence we also have the polynomial representation of $GL_n(A)$ on $A\ot_\K\Xi_n(p)$, which we will also denote by $\rho^{(n; p)}$.
\end{remark}

\subsubsection{Proof of Theorem \ref{thm:paralg-2}}
\begin{proof}
Since $\cD_n(p)\subset \End_\K(\Xi_n(p))$, simplicity of $\Xi_n(p)$ as $\cD_n(p)$-module 
implies $\cD_n(p)=\End_\K(\Xi_n(p))$. Thus we only need to prove the simplicity of $\Xi_n(p)$. 

Note that all $\vartheta_i$ and $\delta_i$ belong to the Lie algebra $\fg$, thus $\cD_n(p)$ is generated by $\fg$ (hence  
is a quotient of the universal enveloping algebra $\U(\fg)$). Thus simplicity of
$\Xi_n(p)$ as a $\U(\fg)$-module is the same as its simplicity as a $\cD_n(p)$-module. 
\end{proof}

\section{Affine para-spaces}\label{sect:affine}

In this section and the next section, we introduce para-spaces and 
investigate basics of differential analysis on them. 
Let us start by studying affine para-spaces. 
The underlying field $\K$ is either $\R$ or $\C$.

\subsection{The affine para-space $\K^{m|n}(p)$}

Fix a non-negative integer $m$,  and positive integers $n$ and $p$.
We construct non-commutative analogues of the affine superspace $\K^{m|n}$ 
by using the para-algebra $\Xi_n(p)$.  We show that these non-commutative spaces 
admit differential analysis similar to that on the affine superspace.

For any open set $U$ in the affine space $\K^m$, we shall consider one of the three families of functions on $U$:  the $C^\infty$, real analytic, and complex analytic (for $\K=\C$) functions, depending on the situation. Adopting notation of \cite[Chapter 1]{W}, we denote by $\CS$ one of the families, and by $\CS(U)$ the ring of $\K$-valued functions belonging to the family $\CS$.  
If $V\subset U$, we have the natural restriction $r^U_V: \CS(U)\lra \CS(V)$. 
Let 
\beq\label{eq:para-SU}
\CS(U, \Xi_n(p))= \CS(U)\ot_\K \Xi_n(p). 
\eeq
This is a non-commutative algebra over the commutative ring $\CS(U)$, which is free as $\CS(U)$-module. 
If $V\subset U$, the map $r^U_V$  leads to the following restriction map $r^U_V\ot_\K \id_{\Xi_n(p)}: 
\CS(U, \Xi_n(p)) \lra \CS(V, \Xi_n(p))$, which is clearly an $\CS(U)$-algebra homomorphism. 
Thus we have the free sheaf $\CS_{\K^m}(n; p)$ of para-algebras consisting of $\CS(U, \Xi_n(p))$ for 
all open subsets $U\subset \K^m$ and the restriction maps

Often we will regard an element $F\in \CS(U, \Xi_n(p))$ as a function $F: U\lra \Xi_n(p)$.

\begin{definition} 
Denote by $\K^{m|n}(p)$ the pair $(\K^m, \CS_{\K^m}(n; p))$, i.e., the affine space $\K^m$ equipped with the free sheaf $\CS_{\K^m}(n; p)$ of para-algebras. Call $\K^{m|n}(p)$ the affine para-space of dimension $m|n$ and order $p$. 
\end{definition} 

Write the coordinate of $\K^m$ as $(x_1, x_2, \dots, x_m)$, and denote the generators of $\Xi_n(p)$ by $\vartheta_i$ for $i=1, 2, \dots, n$ as in Definition \ref{def:Xi}.   We will refer to 
\beq\label{eq:aff-coord}
(x_1, x_2, \dots, x_m, \vartheta_1, \vartheta_2, \dots, \vartheta_n)
\eeq
as the coordinate of $\K^{m|n}(p)$.

Note that $\K^{m|n}(1)=\K^{m|n}$, the affine superspace. 

\subsection{Integral on $\Xi_n(p)$} 
Analysis on $\K^{m|n}(p)$ is to study $\CS(\K^m,  \Xi_n(p)) = \CS(\K^m)\ot_\K \Xi_n(p)$, 
and hence the first step is to investigate $\CS(\K^m)$ and $\Xi_n(p)$
separately. The analysis of $\CS(\K^m)$ is familiar to us, so we focus on the later.  

We have a good candidate for the algebra of ``differential operators'' on $\Xi_n(p)$, namely, $\cD_n(p)$, but we have no integral on $\Xi_n(p)$.  Let us construct one. 

We have constructed an integral on $\ILam_n(p)$ (see Definition \ref{def:int-ILam}), which is shown in Theorem \ref{thm:int-p} to be independent of the embedding of $\ILam_n(p)$ in $\Lam_n(p)$. 
Recall from Corollary \ref{cor:iso-3} the algebra isomorphism $\iota_0: \Xi_n(p)\lra \ILam_n(p)$ (which is the restriction of  the algebra isomorphism $\iota: \cD_n(p)\lra \IW_n(p)$ given inTheorem \ref{thm:iso-2}).  
It allows us to translate the integral on $\ILam_n(p)$ to $\Xi_n(p)$. 
\begin{definition} Define the following map
\beq
\int_{\Xi_n(p)}:= \int_{(p)}\circ\iota_0^{-1}: \Xi_n(p) \lra \K,
\eeq
and  call it an integral on $\Xi_n(p)$.
\end{definition}

It follows Theorem \ref{thm:int-p} that 
\begin{theorem}\label{thm:int-Xi}
For any $F\in \Xi_n(p)$, 
\[
\int_{\Xi_(p)} F = \frac{1}{(p!)^n}\delta_n^p \delta_{n-1}^p \dots \delta_1^p(F). 
\]
\end{theorem}
Using the above theorem, or using properties of $\int_{(p)}$ proved in Appendix \ref{sect:graded}, we obtain the following result. 

\begin{lemma} The integral on $\Xi_n(p)$ can be evaluated for all $\wh\Theta_{i_1, i_2, \dots, i_\ell}$ as follows: 
\[
\int_{\Xi_n(p)}\wh\Theta_{i_1, i_2, \dots, i_\ell}\ne 0 \ \text{  only if $d_j(i_1, i_2, \dots, i_\ell)=p$ for all $j$}.
\]
In particular, 
$\int_{\Xi_n(p)}
\vartheta_1^p \vartheta_2^p\dots \vartheta_n^p=(p!)^n$.  
\end{lemma}

\subsection{Differential analysis on $\K^{m|n}(p)$}
Take an open set $U\subset\K^m$ and consider $\CS(U, \Xi_n(p))$. 
Denote by $\SD(U)$ the ring of differential operators on $U$ with the usual action $\SD(U)\times\CS(U)\lra  \CS(U)$ (which is $\K$-bilinear). Let  
\beq
\SD(U, \Xi_n(p))= \SD(U)\ot_\K \cD_n(p), 
\eeq
whose elements naturally act on $\CS(U, \Xi_n(p))=\CS(U)\ot_\K\Xi_n(p)$.

Now there is a natural $\K$-bilinear action of $\SD(U, \Xi_n(p))$ on $\CS(U, \Xi_n(p))$, 
\beq\label{eq:analysis}
\SD(U, \Xi_n(p))\times \CS(U, \Xi_n(p))\lra \CS(U, \Xi_n(p)).
\eeq
Some kind of ``differentiation'' in the para-variables $\vartheta_i$ is effected by the operators $\delta_i$; 
these operators play an essential role in the differential analysis on $\K^{m|n}(p)$, 
and any para-manifolds to be defined. 
The investigation of the action \eqref{eq:analysis} may be regarded as differential analysis on $(U, \CS(U, \Xi_n(p)))$. 

\begin{remark}\label{rmk:para-diff}
For easy reference, we shall call  $\delta_i$ the {\em para-differential operators} in the context of para-spaces.  
\end{remark}

Observe that the non-commutativity of $\CS(U, \Xi_n(p))$ does not cause any problems of principle; rather it makes the study more interesting.
It is particularly important to note that usual notion of differential equations generalises to the present setting.
For convenience of reference, we shall call equations involving operators in $\SD(U, \Xi_n(p))$ acting on functions in $\CS(U, \Xi_n(p))$ as {\em para-differential equations}.

Para-differential equations can be applied to model physical problems, for example, one could study Schr\"odinger equations on $\K^{m|n}(p)$, generalising the investigations on superspaces in \cite{D1, Z}. 

Below we consider a very simple para-differential equation for the purpose of illustration. 

\begin{example}[A Schr\"{o}dinger equation on para-space]  \label{eg:Schrodinger} 
Write $(x, \vartheta)$ for the  coordinate of $\R^{1|1}(p)$, and denote by $\delta$ the para-differential operator with respect to $\vartheta$.  Let $H\in\SD(\R, \Xi_1(p))$ be given by
\beq
H=- \frac{1}2 \left(\frac{\partial}{\partial x}\right)^2 +\frac{1}2\left( x^2 + 2x \delta + \delta^2\right). 
\eeq
Consider the eigen  para-differential equation 
\beq\label{eq:para-DE}
H F(x, \vartheta) = E F(x, \vartheta), 
\eeq
where $E\in\R$ is the eigenvalue, and $F(x, \vartheta)\in \CS(\R, \Xi_1(p))$ with $\CS$ being the $C^\infty$ family.

Let $H_n(x)$ be the $n$-th Hermite polynomial, and define 
\[
\CH_n(x, \delta)= \sum_{\ell\ge 0} \frac{\delta^\ell}{\ell!}  \left(\frac{d}{d x}\right)^\ell H_n(x), \quad 
F_0(x, \delta)= \sum_{\ell\ge 0} \frac{\delta^\ell}{\ell!}  \left(\frac{d}{d x}\right)^\ell e^{-x^2/2},  
\]
both of which are finite sums as $\delta^{p+1}=0$. 
Given any non-zero $f(\vartheta)\in \Xi_1(p)$, we let 
$
F_n(x, \vartheta) = \CH_n(x, \delta) F_0(x, \delta)(f(\vartheta)), 
$
which is an element of $\CS(\R, \Xi_1(p))$. 
Then  
\[
H F_n(x, \vartheta) = \left(n + \frac{1}{2}\right) F_n(x, \vartheta), \quad \forall n=0, 1, \dots,
\]
i..e, $F_n(x, \vartheta)$ is an eigenfunction of $H$ with eigenvalue $n + \frac{1}{2}$. 

\begin{remark}\label{rmk:Sch}
If we think of $H$ as a Hamiltonian operator and \eqref{eq:para-DE} as a Schrodinger equation, the different $F_n(x, \vartheta)$ corresponding to different choices of $f(\vartheta)$ are the degenerate states at the energy $n + \frac{1}{2}$ (in some unit). 
\end{remark}
\end{example}

\subsection{Integral on $\K^{m|n}(p)$}\label{sect:int-affine}

Let us construct an integral, which maps appropriate elements of $\CS(U, \Xi_n(p))$ to scalars. Such an integral is indispensable for doing analysis on para-spaces. 

As a $\K$-vector space, $\Xi_n(p)$ is finite dimensional. We choose a basis $\{B_a\mid a=1, 2, \dots, \dim\}$ for it, where $\dim=\dim_\K(\Xi_n(p))$. Any $F\in   \CS(U, \Xi_n(p))$ can be expressed as $F=\sum_a f_a\ot B_a$ with $f_a\in\CS(U)$.  Let $ \CS(U, \Xi_n(p))^C$ be the subspace consisting of $F\in  \CS(U, \Xi_n(p))$ such that all $f_a$ are integrable, that is $\int_{U} f_a$ exist for all $a$. Then we have the map  $\int_{U}\ot\id:  \CS(U, \Xi_n(p))^C\lra \Xi_n(p)$.  Compose this with the integration $\int_{(p)}$ on the para-algebra, we obtain the map
\beq\label{eq:integration}
\int:=\int_{\Xi_n(p)}\circ\left(\int_{U}\ot\id_{\Xi_n(p)}\right):  \CS(U, \Xi_n(p))^C\lra \K.  
\eeq
This defines an integral on $(U,  \CS(U, \Xi_n(p)))$.  

\begin{remark}
One may take $\int=\int_{U}\circ\left(\id_{\CS(U)}\ot\int_{\Xi_n(p)}\right)$ instead, then $\int F$ for an element $F\in \CS(U, \Xi_n(p))$ is well defined so long as the coefficient (belonging to $\CS(U)$) of $\Theta^{\SE_{max}}$ in $F$ is integrable over $U$.
\end{remark}

\section{Para-manifolds}\label{sect:folds}

We consider more complicated non-commutative geometries by gluing together affine para-spaces, which will be referred to as para-manifolds. We then develop in some detail a family of examples denoted by $\BKP^{m|n}(p)$.  

The development of a general theory of para-manifolds will be pursued in a future publication \cite{Z23-b}.

\subsection{Para-manifolds}
We adopt some terminology from \cite{W}. 

Fix a non-negative integer $m$,  and positive integers $n$ and $p$. 
Let $X$ be a connected $\CS$-manifold of dimension $m$ over $\K$, whose structure sheaf will be denoted by $\CS_X$.

A sheaf $\CA$ of associative algebras (which is not necessarily commutative) on $X$ is 
a  sheaf of $\CS_X$-modules such that 
$\CS_X\subset \CA$ is a subsheaf, 
$\CA(U)$ on any open set $U\subset X$ is an $\CS_U$-algebra, and 
for any $V\subset U$, the restriction map to from $U$ to $V$ is an $\CS_U$-algebra homomorphism. 
The sheaf of algebras is said to be $\Z_+$-graded if 
$\CA=\oplus_{k\ge 0}\CA_k$ for subsheaves $\CA_k$ such that 
$\CA_0=\CS_X$, and 
$\CA_k(U)\CA_\ell(U)\subset\CA_{k+\ell}(U)$ for all $k, \ell$ on any open $U\subset X$. 

We take the following as a``working definition'' of sheave of para-algebras.
\begin{definition} 
Let $X$ be a connected $\CS$-manifold. Let $\CS_X(n;p)$ be a sheaf of algebras on $X$ with the following properties:
\begin{enumerate}[i).]
\item
there exists a locally free subseaf $\CE_X$ of vector spaces of rank $n$,  
\item for any $x\in X$, there is an open neighbourhood $B\ni x$ such that $\CS_X(B, n; p)$ is generated by $\CE_X(B)$ as $\CS_X(B)$-algebra,
and the following relations hold
\beq
\xi_1^p \xi_2=- \xi_2 \xi_1^p, \quad
[\xi_1, [\xi_2, \xi_3]]=0, \quad \forall\xi_1, \xi_2, \xi_2\in \CE_X(B).
\eeq
%
%
\end{enumerate}
Call $\CS_X(n;p)$ a sheaf of para-algebras of order $p$ and degree $n$.
\end{definition} 
We say that $\CS_X(n;p)$ is generated by $\CE_X$. 

We want to study manifolds equipped with such sheaves of para-algebras; 
this is inspired by the algebraic geometric formulation of supermanifolds \cite{L}. 
For easy reference, we introduce the following terminology.
\begin{definition} Let $X$ be a connected $\CS$-manifold, and let $\CS_X(n; p)$ be a sheaf of para-algebras on $X$.  
Call the ringed space $(X, \CS_X(n; p))$ a {\em para-manifold} and denote it 
by $\X(n; p)$.  
\end{definition} 

We will refer to $\CS_X(n; p)$ as the structure sheaf, $X$ the base manifold, 
$p$ the order and $m|n$ the dimension,  of $X(n; p)$. 

\begin{remark}
Observe that $\X(1)$ is a supermanifold  of super dimension $m|n$. 
\end{remark}

As $\CE_X$ is a locally free sheaf (a vector bundle),  
there is an open covering $\SC=\{U_\alpha\}$ of $X$  such that 
$\CE_X(U_\alpha)\cong  \CS_X(U_\alpha)^{\oplus n}$ for all $\alpha$. Note that $n$ is independent of $\alpha$ as the manifold is connected.  If $U_{\alpha\beta}:=U_\alpha\cap U_\beta\ne \emptyset$, there exists an isomorphism of $\CS_X(U_{\alpha\beta})$-modules 
\[
g_{\alpha \beta}:  \CE_X(U_\alpha)|_{U_{\alpha\beta}}\lra \CE_X(U_\beta)|_{U_{\alpha\beta}}, \quad g_{\alpha \beta} \in \GL_n(\CS_X(U_{\alpha\beta})).
\]
Now by the definition of $\CS_X( n; p)$ and the discussions in Section \ref{sect:GL-mod}, we have 
\[
\CS_X(U_\alpha, n; p)\cong \CS_X(U_\alpha)\ot_\K \Xi_n(p), \quad
\CS_X(U_\beta, n; p)\cong \CS_X(U_\beta)\ot_\K \Xi_n(p), 
\]
and the isomorphisms of $\CS_X(U_{\alpha\beta})$-modules 
\beq
&&\rho^{(n; p)}(g_{\alpha \beta} ): \CS_X(U_\alpha, n; p)|_{U_{\alpha\beta}}\lra \CS_X(U_\beta, n; p)|_{U_{\alpha\beta}}, 
\eeq
where $\rho^{(n; p)}$
is direct sum of polynomial representations of $\GL_n(\CS_X(U_{\alpha\beta}))$ defined in Remark \ref{rmk:GL-rep}.

The theorem below summarises the discussions above. 
\begin{theorem}\label{thm:split}
The para-manifold $\X(n; p)$ is split in the sense that its structure sheaf $\CS_X(n; p)$  is generated by $\CE_X$ and is locally free. 
\end{theorem}
This is an analogue of Bechalor's Theorem for $C^\infty$ supermanifolds \cite{B}. 

\subsection{Differential analysis}
Differential analysis of $\CS(U, n;p)$ for any $U$ in the open cover $\SC$ of $X$ is the same as in the case of affine para-spaces, and  was already investigated in Section \ref{sect:affine}.  

To extend the local differential analysis to the whole para-manifold, we consider two coordinate charts 
$$
(U, (x^i,\phi^\mu \mid 1\le i\le m, \ 1\le  \mu\le n)), \quad (V, (y^i, \psi^\mu \mid 1\le i\le m, \ 1\le  \mu\le n)),
$$
where both $U, V$ belong to $\SC$ and $U\cap V\ne\emptyset$.  For the first chart, the $x^i$ are coordinates of $U$ (through some homemorphsism to $\K^m$ and we will be suppress the map in the notation for convenience), and $\phi^\mu$ are the elements of a basis of $\CE(U)$ as a free $\CS(U)$-module, 
and similarly for the other chart. 
[Here we have switched to the usual convention in differential geometry that coordinate components are labelled by upper indices; that helps us to keep track of changes of coordinates.]

The following is the most general transformation rule for coordinates,
\beq\label{eq:coord-transf}
y^i=Y^i(x),  \quad  \psi^\mu =\sum_\nu g(x)^\mu_\nu \phi^\nu,  
\eeq
where the map $Y: U\lra V$ is invertible and  $g(x)$ is a non-sngular matrix at every point. 

Denote by  $\delta_{\phi^\mu}$ and $\delta_{\psi^\mu}$ the para-differential operators with respect to 
$\phi^\mu$ and $\psi^\mu$. 
We have the following result. 
\begin{theorem} \label{thm:transf-oper} 
Retain the setting above; in particular,  the two sets of coordinates of $U\cap V$  transform according to \eqref{eq:coord-transf}.  
Then the consistent (the meaning is explained in the proof) transformations between partial differential operators and para-differential operators on $U\cap V$ are given by
\beq
&&\frac{\partial}{\partial x^i} =\sum_j J(x)^j_i \frac{\partial}{\partial y^j} -  \frac{1}{2} \sum_{\mu, \nu} {G_i(x)}^\mu_\nu\left( [\psi^\nu, \delta_{\psi^\mu}] + p \delta^\nu_\mu\right), \label{eq:dx-dy}\\
&&\delta_{\phi^\mu}= \sum_{\nu} g(x)^\nu_\mu \delta_{\psi_\nu}, \label{eq:dphi-dpsi} \\
\text{with} && J(x)^i_j=\frac{\partial Y^i(x)}{\partial x^j}, \quad  {G_i(x)}^\mu_\nu =\frac{\partial ((g(x))^{-1})^\sigma_\nu}{\partial x^i}  g(x)^\mu_\sigma.
\eeq
\end{theorem}
\begin{proof}
Consistency means the following. 
Take any three coordinate charts,     
\[
(U, (x^i,\phi^\nu)), \quad (V, (y^i, \psi^\mu)), \quad (W, (z^i, \eta^\mu)), 
\]
where $U, V, W\in\SC$ and $U\cap V\cap W\ne \emptyset$. Assume that
\beq
&&y^i=Y^i(x), \quad z^i = Z^i(y), \quad x^i=X^i(z), \\
&& \psi^\mu =\sum_\nu g(x)^\mu_\nu \phi^\nu, 
\quad \eta^\mu =\sum_\nu \wt{g}(y)^\mu_\nu \psi^\nu, 
\quad \phi^\mu =\sum_\nu \wt{\wt{g}}(z)^\mu_\nu \eta^\nu, 
\eeq
which satisfy the following consistency requirements,
\beq
X^i(Z(Y(x))=x^i, \label{eq:XYZ} \\
 g(x) \wt{g}(y) \wt{\wt{g}}(z)  =1,  \label{eq:ggg}
\eeq
where $g(x)$ etc. are the following $n\times n$ matrices, 
\[
g(x)= \left( g(x)^\mu_\nu\right), \quad \wt{g}(y)= \left( \wt{g}(y)^\mu_\nu\right), \quad \wt{\wt{g}}(z)= \left( \wt{\wt{g}}(z)^\mu_\nu\right). 
\] 
We apply the transformation rules \eqref{eq:dx-dy} and \eqref{eq:dphi-dpsi} for partial derivatives and para-differential operators three times, first from $U$ to $V$, then from $V$ to $W$, and finally from $W$ to $U$, 
we should get an identity transformation.  
As the para-differential operators are not derivations,  this is by no means guaranteed.

To prove the theorem, it is convenient to introduce the matrices 
\[
\baln
& G_i =\frac{\partial g^{-1}}{\partial x^i} g, \quad
\wt{G}_i =\frac{\partial \wt{g}^{-1} }{\partial y^i}  \wt{g}, \quad
 \wt{\wt{G}}_i =\frac{\partial \wt{\wt{g}}^{-1} }{\partial  z^i} \wt{\wt{g}}; \\
& J=\left(\frac{\partial Y^i(x)}{\partial x^j}\right), \quad \wt{J}= \left(\frac{\partial Z^i(y)}{\partial y^j}\right), 
\quad \wt{\wt{J}}= \left(\frac{\partial X^i(z)}{\partial z^j}\right), 
\ealn
\]
and arrange $\phi^\mu$ into a row vector and $\frac{\partial}{\partial x^i}$ and $\delta_{\phi^\mu}$ into column vectors, and similarly for those objects on the other charts.
Then we can write the transformation rules in matrix form as $\psi=\phi g$,  $\eta=\psi \wt{g}$, $\phi=\eta \wt{\wt{g}}$, and 
\beq
\delta_{\phi}= g\delta_{\psi},  \quad
 \frac{\partial}{\partial x} = J\frac{\partial}{\partial y} -  \frac{1}{2}\left( \sum_\mu [\psi^\mu, (G_*\delta_{\psi})_\mu] + p tr(G_*) \right), \label{eq:U-V}\\
 \delta_{\psi}= \wt{g}\delta_{\eta},  \quad
 \frac{\partial}{\partial y} = \wt{J}\frac{\partial}{\partial z} 
-  \frac{1}{2}\left( \sum_\mu [\eta^\mu, (\wt{G}_*\delta_{\eta})_\mu] + p tr(\wt{G}_*) \right),   \label{eq:V-W} \\
 \delta_{\eta}= \wt{\wt{g}}\delta_{\phi},  \quad
 \frac{\partial}{\partial z} = \wt{\wt{J}}\frac{\partial}{\partial x} 
-  \frac{1}{2}\left( \sum_\mu [\phi^\mu, (\wt{\wt{G}}_*\delta_{\phi})_\mu] + p tr(\wt{\wt{G}}_*) \right),  \label{eq:W-U}
\eeq
where the $_*$'s in three of  the equations above indicate that we arrange $tr(G_i)$ etc. into column vectors. 

Now we have 
$
\phi= \phi g \wt{g} \wt{\wt{g}}, 
$
which is indeed an identity by \eqref{eq:ggg}.

Also, use the second relation of \eqref{eq:V-W} in the second equation of \eqref{eq:U-V} to eliminate $\frac{\partial}{\partial y}$, 
then use the second relation of \eqref{eq:W-U} in the resulting equation, we obtain 
\beq\label{eq:x-x}
\frac{\partial}{\partial x} = J\wt{J}\wt{\wt{J}}\frac{\partial}{\partial x}   - \frac{N}{2}, 
\eeq
with 
\[
\baln
 N=
& \sum_\mu [\phi^\mu, (J\wt{J}\wt{\wt{G}}_*\delta_{\phi})_\mu] + p tr(J\wt{J}\wt{\wt{G}}_*) \\
&+\sum_\mu [\eta^\mu, (J\wt{G}_*\delta_{\eta})_\mu] + p tr(J\wt{G}_*)\\
&+\sum_\mu [\psi^\mu, (G_*\delta_{\psi})_\mu] + p tr(G_*). 
\ealn
\]
We want to show that the right hand side of \eqref{eq:x-x} is equal to $\frac{\partial}{\partial x}$. As $J\wt{J}\wt{\wt{J}}=1$, this requires showing $N=0$.  

Let 
$
Z_i= J\wt{J}\wt{\wt{G}}_i + \wt{\wt{g}}^{-1}J\wt{G}_i \wt{\wt{g}} +  g G_i g^{-1}.
$
Then 
\[
N_i=\sum_\mu [\phi^\mu, (Z_i \delta_{\phi})_\mu] + p tr(Z_i), \quad \forall i. 
\]
Note that 
$
J\wt{J}\wt{\wt{G}}_i = \frac{\partial \wt{\wt{g}}^{-1} }{\partial  x^i} \wt{\wt{g}}$ and $J\wt{G}_i = \frac{\partial \wt{g}^{-1} }{\partial x^i}  \wt{g}.
$
Thus 
\[
\baln
Z_i 
&= \frac{\partial \wt{\wt{g}}^{-1} }{\partial  x^i} \wt{\wt{g}} 
+  \wt{\wt{g}}^{-1}   \frac{\partial \wt{g}^{-1} }{\partial x^i}  \wt{g}   \wt{\wt{g}}  
+    g \frac{\partial g^{-1}}{\partial x^i}.
\ealn
\]
The middle term can be re-written as 
$ \frac{\partial (\wt{g}\wt{\wt{g}})^{-1}  }{\partial x^i}  \wt{g}   \wt{\wt{g}}  
-  \frac{\partial \wt{\wt{g}}^{-1} }{\partial x^i}  \wt{\wt{g}}$ 
by some very simple algebraic manipulations. Replacing $\wt{g}   \wt{\wt{g}}$ by $g^{-1}$ by using  \eqref{eq:ggg}, we can further re-write this expression as $ \frac{\partial g }{\partial x^i}  g^{-1}  
-  \frac{\partial \wt{\wt{g}}^{-1} }{\partial x^i}  \wt{\wt{g}}$. Hence 
\[
Z_i =  \frac{\partial g }{\partial x^i}  g^{-1} + g \frac{\partial g^{-1}}{\partial x^i}=0.
\]
Hence $N=0$, completing the proof of the theorem. 
\end{proof}

\begin{remark}
Denote the right hand sides of equations \eqref{eq:dx-dy} and \eqref{eq:dphi-dpsi} by $\partial_i$ and $\delta_\mu$ respectively. Then in the overlapping region of the two charts,  
\[
\partial_i(x^j)=\delta_i^j, \quad \delta_\mu(\phi^\nu)=p \delta^\nu_\mu, \quad \partial_i(\phi^\mu)=0, \quad \delta_\mu(x^i)=0.
\]
\end{remark}

\begin{remark}\label{rmk:D-mod}
A more sophisticated version of global analysis on $\X(p)$ can be developed by building 
a sheaf of rings analogous to the sheaf of differential operators on supermanifolds, 
and then analysing the sheaves of modules for it. 
\end{remark}

\begin{remark}
A DeWitt type formulation \cite{dW} for para-manifolds will need to involve $\Xi_\infty(p)$ (direct limit of $\Xi_n(p)$), the order $p$ para-algebra of infinite degree.  
\end{remark}

It is beyond the scope of this note (nor is it our intention) to develop a comprehensive theory of para-manifolds. 
Below we treat in depth a family of simple non-trivial examples to illustrate some features of the theory.

\subsection{Examples -  projective para-spaces $\BKP^{m|n}(p)$}\label{sect:CPmn}

We consider some simplest examples of para-spaces 
with base manifolds being projective spaces, and develop basics of their differential analysis.  
We start by looking at cases where the base manifold is the projective line $\BKP^1$.

\subsubsection{$\BKP^{1|n}(p)$}\label{sect:CP1n}

 Denote by $\CO$ the structure sheaf of the projective line $\BKP^1$. 
As usual, we cover $\BKP^1$ with two coordinate charts $U=\{z\in\K\}$ and $V=\{w\in\K\}$ such that 
$w=\frac{1}z$ on the intersection $U\cap V$. 
Let
\beq\label{eq:O-sub-n}
\CO(U, n; p)=\CO(U)\ot_\K \Xi_n(p), \quad \CO(V, n;p)=\CO(U)\ot_\K \Xi'_n(p), 
\eeq
where $\Xi_n(p)$ is generated by $\phi^1,  \phi^2, \dots, \phi^n$ subject to the relations \eqref{eq:the-p-0} and \eqref{eq:R-0} (but with $\vartheta_i$ replaced by $\phi^i$), and $\Xi'_n(p)$ is similarly generated by $\psi_1,  \psi_2, \dots, \psi_n$. We have two coordinate charts of $\BKP^{1|n}(p)$ with coordinates as given below
\[
\baln
 \left(U, \CO(U, n; p)\right): \quad & (z, \phi^1, \phi^2, \dots, \phi^n); \\
 \left(V, \CO(V, n; p)\right): \quad & (w, \psi^1, \psi^2, \dots, \psi^n).  
\ealn
\]
In the region $U\cap V$, we have the following coordinate transformation
\beq
w=\frac{1}z, \quad \psi^i=\frac{\phi^i}z,   \quad i=1, 2, \dots. \label{eq:coord-n}
\eeq

Corresponding to $\Xi_n(p)$ and $\Xi'_n(p)$,  we have the para Clifford algebras $\cD_n(p)$ and and $\cD'_n(p)$, which are generated by 
 $\{\phi^i, \delta_{\phi^i}\mid i=1, 2, \dots, n\}$, and  
 $\{\psi_i, \delta_{\psi_i}\mid i=1, 2, \dots, n\}$ respectively. Thus we have the algebras
 \[
\SD_U(n;p)= \cD_U\ot_\K \cD_n(p), \quad \SD_V(n;p)= \cD_V\ot_\K \cD'_n(p), 
\] 
where $\cD_U$ (resp. $\cD_V$) is the  algebra of differential operators on $U$ (resp. $V$).  
The algebra $\SD_U(n;p)$ (resp. $\SD_V(n;p)$) naturally acts on $\CO(U, n; p)$ (resp. $\CO(V, n; p)$).  
In the region $U\cap V$, the differential operators and para-differential operators transform according to the following rules.
\beq
&&\delta_{\psi^i} = z \delta_{\phi^i}, \quad  i=1, 2, \dots,   \label{eq:del-n}\\
&&\frac{\partial}{\partial w} = -z^2 \frac{\partial}{\partial z} + \frac{z}{2} \sum_{i=1}^n\left([\phi^i, \delta_{\phi^i}] + p\right). \label{eq:part-rule-n}
\eeq

\subsubsection{Projective para-spaces $\BKP^{m|n}(p)$}
Let us generalise the above construction to general projective para-spaces $\BKP^{m|n}(p)$. 
This will also give some context to the construction in Section \ref{sect:CP1n}.

Consider $\K^{m+1|n}(p)$ with coordinate
\beq
(x^0, x^1, x^2, \dots, x^m; \vartheta^1, \vartheta^2, \dots, \vartheta^n),
\eeq
where $(x^0, x^1, x^2, \dots, x^m)$ is the coordinate of $\K^{m+1}$. 

We take the following open subsets 
$
V_\alpha =\{x=(x^0, x^1, x^2, \dots, x^m)\in \K^{m+1}\mid x^\alpha\ne 0\}$ of $\K^{m+1|n}$,  
for $ \alpha=0, 1, \dots, m$. For each $V_\alpha$,  we define 
$
z^i_\alpha = \frac{x^{i-1}}{x^\alpha}$ for $1\le  i\le \alpha$, and $z^j_\alpha = \frac{x^j}{x^\alpha}$ for $\alpha+1\le  j\le m. 
$
Then we have the following open set of $\K^m$.
\[
U_\alpha =\{(z^1_\alpha, z^2_\alpha, \dots, z^m_\alpha)\mid  x\in V_\alpha\}, \quad \alpha=0, 1, \dots, m.
\]
On $U_\alpha\cap U_\beta$ for $\alpha < \beta$, the coordinates transform as follows. 
\beq
&&z^i_\beta =(z^\beta_\alpha)^{-1} z^i_\alpha, \quad 1\le i\le \alpha,  \label{eq:proj-0}\\
&&z^{\alpha+1}_\beta = (z^\beta_\alpha)^{-1}, \\
&&z^j_\beta = (z^\beta_\alpha)^{-1} z^{j-1}_\alpha, \quad \alpha+1< j \le \beta, \\
&&z^k_\beta = (z^\beta_\alpha)^{-1} z^k_\alpha, \quad \beta< k\le m. \label{eq:proj-3}
\eeq
These enable us to glue the $U_\alpha$'s together to obtain the projective space $\BKP^{m}$. 

Let $\vartheta^\mu$ with $1\le\mu\le n$ be the generators of $\Xi_n(p)$.
Corresponding to each $U_\alpha$, we define $\phi^\mu_\alpha = \frac{\vartheta^\mu}{x^\alpha}$ for all $\mu=1, 2, \dots, n$, which generate the para-algebra ${\Xi_n(p)}_{(\alpha)}\cong \Xi_n(p)$.  
On $U_\alpha\cap U_\beta$ with $\alpha < \beta$, 
\beq
\phi^\mu_\beta  =(z^\beta_\alpha)^{-1}  \phi^\mu_\alpha, \quad  \mu=1, 2, \dots, n. \label{eq:proj-4}
\eeq

Let $\CO(U_\alpha, n; p) =\CO(U_\alpha)\ot_\K \Xi_n(p)_{(\alpha)}$, and introduce the following 
coordinate charts with local coordinates as indicated. 
\[
\baln
\left(U_\alpha, \CO(U_\alpha, n; p)\right): \  (z^1_\alpha, z^2_\alpha, \dots, z^m_\alpha, \phi^1_\alpha, \phi^2_\alpha, \dots, \phi^n_\alpha), \quad \alpha=0, 1, \dots, m. 
\ealn
\]
The coordinate transformation rules \eqref{eq:proj-0} - \eqref{eq:proj-4} enable us to glue the charts together to construct a para-manifold $\left(\BKP^{m|n}, \CO_{\BKP^{m|n}}(n; p)\right)$. We shall denote it by $\BKP^{m|n}(p)$, and call it a projective para-space. 

The transformation rules between the operators $\frac{\partial}{\partial z^i_\alpha}$, $\delta_{\phi^\nu_\alpha}$ and $\frac{\partial}{\partial z^i_\beta}$, $\delta_{\phi^\nu_\beta}$ in the region $U_\alpha\cap U_\beta$ are give by Theorem \ref{thm:transf-oper} in terms of explicit formulae, 
which we omit

\subsection{Comments}\label{sect:differs}
We pointed out in Section \ref{sect:intro} that para-manifolds are not $\Z_2^k$-graded commutative in contrast to
$\Z_2^k$-supermanifolds. 
There are also major technical differences between para-manifolds and $\Z_2^k$-supermanifolds. 
In the simplest case of an affine $\Z_2^2$-superspace, the algebra of $\CS$-functions on an open set $U$ consisting of such functions $U\lra Z$, where $Z$ is  
a $\Z_2^2$-graded commutative algebra generated by $3$ generators $\zeta_1, \zeta_2, y$,  
which have $\Z_2^2$-degrees 
$(1, 0)$, $(0, 1)$ and $(1, 1)$ respectively, with defining relations 
\[
\baln
\zeta_1^2= \zeta_2^2=0, \quad \zeta_1\zeta_2 = \zeta_2 \zeta_1,  \quad \zeta_i y + y  \zeta_i=0, \  i=1, 2. 
\ealn
\]
Note that $y$ can not be assigned scalar values as it anti-commutes with $\zeta_1$ and $\zeta_2$. 
Also since $y$ is not nilpotent, one allows \cite{AI, Po} formal power series of $y$ in $Z$,  
\[
Z= \K[[y]] + \sum_{i=1, 2} \K[[y]] \zeta_i + \K[[y]] \zeta_1\zeta_2. 
\]
It is quite evident that ``integration'' over $y$ is going to be a thorny issue, as $y$ can not be  treated as a usual variable, and trying to get finite scalars from formal power series of $y$ is also going to be problematic.  We refer to \cite{AI, Po} for suggestions of definitions,  and discussions on related problems, of integrals on  $\Z_2^k$-superspaces.  

In contrast, the definition of an integral on an affine para-space is straightforward, 
as we have shown in Section \ref{sect:int-affine}.

\section{Para-space realisation of parafermions}\label{sect:quantise}

We illustrate how Green's parafermions naturally fit into the framework of para-spaces.  
We will consider the case of finitely many degrees of freedom only for technical simplicity, 
that is, within the context of quantum mechanics.

We set $\K=\C$ in this section.

Fix a positive integer $n$. Let $G_n$ be the associative algebra \cite{G, BG} generated by the generators $\Psi_i, \ol\Psi_i$ for $i=1, 2, \dots, n$,  subject to the following defining relations for all $i, j, k$:
\beq
&& [\Psi_k, [\Psi_i, \Psi_j]]=0, \quad [\Psi_k, [\ol\Psi_i, \Psi_j]]=2\delta_{k i}\Psi_j, \label{eq:def-1}\\
&& [\ol\Psi_k, [\ol\Psi_i, \ol\Psi_j]]=0,  \quad 
[\ol\Psi_k, [\ol\Psi_i, \Psi_j]]= -2\delta_{ k j} \ol\Psi_i, \label{eq:def-3}
\eeq
and the following relations
\beq
&& [\ol\Psi_k, [\Psi_i, \Psi_j]] = 2\delta_{k i} \Psi_j - 2\delta_{k j} \Psi_i,    \label{eq:def-2}, \\
&&  [\Psi_k, [\ol\Psi_i, \ol\Psi_j]]=2\delta_{k i}\ol\Psi_j-2\delta_{k j} \ol\Psi_i, \label{eq:def-4}, 
\eeq
where $[X, Y]=X Y - YX$ is the usual commutator in an associative algebra. 
We point out that  \eqref{eq:def-2} and \eqref{eq:def-4} are consequences 
of the second relation of  \eqref{eq:def-1} and the second relation of  \eqref{eq:def-3} respectively. 

\begin{remark}
The algebra $G_n$ is known to be isomorphic to the universal enveloping algebra of the orthogonal Lie algebra $\fso_{2n+1}(\K)$ \cite{BG}. See also \cite{P, LSV} for treatments of parafermions and parabosons from a Lie theoretical point of view. 
\end{remark}

Fix a positive integer $p$.  Let $ F_n(p)$ be the simple $G_n$-module cyclically generated by a vector $\ket0$, called the vacuum vector,  such that 
\beq\label{eq:vacu}
\Psi_i\ket0=0, \quad \Psi_i\ol\Psi_j \ket0 =p\delta_{i j} \ket0, \quad \forall i, j.  
\eeq
This module is known to be finite dimensional \cite{G}.  A basis for $F_n(p)$ was constructed using Gelfand-Zeitlin patterns in \cite{SV08}.

Green's definition \cite{G} of $n$ parafermions of order $p$ refers to the
datum $(G_n,  F_n(p))$ with an $\ast$-structure for $G_n$ and unitarisability of $F_n(p)$. The unitarisability of $F_n(p)$  is
intimately connedcted with the well known unitarisability 
of the finite dimensional simple $\fso_{2n+1}(\R)$-module described in Lemma \ref{lem:g-mod}.

\begin{remark}
The paraboson Fock space is closely related to a unitarisable module for $\osp_{1|2n}(\R)$ \cite{DZ}, see  \cite{LSV} for details. 
\end{remark}

Now let us  realise the parafermion algebra $G_n$ and the module $F_n(p)$ in terms of para-spaces.

In the language of quantum mechanics, we work with the Schrödinger picture,  in which the operators are constant, and the states evolve in time. Operators are obtained by quantising coordinates and conjugate momenta. 
Thus we only need to consider the para-space with a single point $pt$, that is,  
$\K^{0|n}(p)=(\{pt\}, \CS_{pt}(n; p))$, whose coordinate will be written as $(\vartheta_1, \vartheta_2, \dots, \vartheta_n)$. The corresponding para-differential operators will be denoted by $\delta_i:=\delta_{\vartheta_i}$ for $1\le i\le n$. 

We have the following result. 
\begin{theorem}\label{thm:p-fermi} 
Retain notation above and notation in Definition \ref{def:ICliff}. 
There exists an algebra epimorphism $\FB_n^p: G_n\lra \SD(pt, \Xi_n(p))$ such that 
$\CS_{pt}(n; p)$ as a $G_n$-module via $\FB_n^p$ is isomorphic to $ F_n(p)$.  
\end{theorem}
\begin{proof}
We have $\CS_{pt}(n; p)=\Xi_n(p)$ and $\SD(pt, \Xi_n(p))=\cD_n(p)$. 
Recall that $\Xi_n(p)$, which is generated by $\vartheta_1, \vartheta_2, \dots, \vartheta_n$,  satisfies the relations \eqref{eq:the-p-0} and \eqref{eq:R-0}: also $\cD_n(p)$, which is generated by $\vartheta_i$ and $\delta_i$ for $1\le i\le n$, satisfies the relations \eqref{eq:R-1} -- \eqref{eq:R-4}. Comparing them with the defining relations \eqref{eq:def-1} -- \eqref{eq:def-4} of $G_n$, 
we immediately see that there is a unique algebra epimorphism $\FB_n^p: G_n\lra \SD(pt, \Xi_n(p))$ extending the map 
\[
\ol{\Psi}_i\mapsto \vartheta_i, \quad \Psi_i\mapsto \delta_i, \quad  \forall i=1, 2, \dots, n. 
\]

Now $G_n$ naturally acts on $\Xi_n(p)$ via the map. As $\FB_n^p$ is surjective, and  $\Xi_n(p)$ is simple as an $\cD_n(p)$-module by Theorem \ref{thm:paralg-2}, we conclude that $\Xi_n(p)$ is a simple $G_n$-module. Clearly $\Xi_n(p)$ is cyclically generated by $1$ as $\cD_n(p)$-module, and hence also as $G_n$-module.

Note that $1\in \Xi_n(p)$ satisfies the relations \eqref{eq:vac-1} with respect to the $\cD_n(p)$-action.   The relations can be re-written as 
\beq\label{eq:BPsi}
\FB_n^p(\Psi_i)(1)=0, \quad 
\FB_n^p(\Psi_i\ol\Psi_j)(1)=p\delta_{i j}, \quad \forall i, j. 
\eeq
Now we compare this with \eqref{eq:vacu}. Since $F_n(p)$ is a simple $G_n$-module cyclically generated by the vacuum vector $\ket0$, we immediately see that there is a $G_n$-module isomorphism $\Xi_n(p)\lra F_n(p)$ determined by $1\mapsto \ket0$. 
\end{proof}

This shows that para-spaces indeed provide a natural framework for para-fermions.

\appendix
\section{A class of $\Z_2^p$-graded commutative algebras}  \label{sect:graded}
We make use of a class of $\Z_2^p$-graded commutative algebras in the main body of the paper, 
which are treated here.   

We point out that in essence, the inhomogeneous algebra $\ILam_n(p)$ 
defined in Section \ref{sect:inhom}  was discussed by Ohnuki and Kamefuchi in \cite{OK80-1, OK80} 
and subsequent papers, where elements of $\sum_i\K\theta_i$ were called ``para-Grassmann numbers''.
The $\Z_2^p$-graded commutative algebra $\Lam_n(p)$ is 
what they called a generalised Grassmann algebra. However in loc. cit. the notion of $\Z_2^p$-grading was not made explicit, even though it was used. 

\subsection{The $\Z_2^p$-graded commutative algebra $\Lambda_n(p)$}
Recall that the Grassmann algebra $\Lam_n(1)$  of degree $n$ over a field $\K$ is a $\Z_2$-graded associative algebra generated by the odd elements $\zeta_i$ for $1\le i\le n$, which satisfy the standard relations
\[
\zeta_i \zeta_j = - \zeta_j\zeta_i, \quad \forall i, j. 
\]
There exist derivations $\frac{\partial}{\partial \zeta_i}: \Lam_n(1)\lra \Lam_n(1)$, which are odd (i.e., homogeneous of degree $1$) defined by 
\beq
&&\frac{\partial}{\partial \zeta_i}(A B) = \frac{\partial}{\partial \zeta_i}(A) B + (-1)^{deg(A)} A  \frac{\partial}{\partial \zeta_i}(B),   \quad A, B\in \Lam_n(1), \\
&&\frac{\partial}{\partial \zeta_i}(\zeta_j)=\delta_{i j}, \quad \forall i, j.
\eeq
It is essential for consistency that the derivations $\frac{\partial}{\partial \zeta_i}$ are odd.   

Let us generalise the above construction to the $\Z_2^p$-graded setting for any given finite positive integer $p$. 

The elements of the abelian group $\Z_2^p$  will be written as $p$-tuples with entries in $\Z_2$. In particular, we denote
$
{\bf e}_\alpha=(\underbrace{0,\dots, 0}_{\alpha-1}, 1, 0 \dots, 0)$ for $1\le\alpha\le p.
$
Then elements of $\Z_2^p$ are sums of ${\bf e}_\alpha$'s. 
Introduce a function $(\ , \ ): \Z_2^p\times \Z_2^p\lra \Z_2$ defined by $(\lambda, \mu)=\sum_{\alpha=1}^p \lambda_\alpha\mu_\alpha$ for any $ \lambda, \mu\in \Z_2^p$.   Note that $({\bf e}_\alpha, {\bf e}_\beta)=\delta_{\alpha\beta}$.

A $\Z_2^p$-graded associative algebra $\CA$ is a $\Z_2^p$-graded vector space $\CA=\sum_{\lambda\in\Z_2^p} \CA_\lambda$ over $\K$ with a bilinear associative multiplication $\CA\times \CA\lra \CA$  such that $\CA_\lambda \CA_\mu\subset \CA_{\lambda+\mu}$.

Given any $\Z_2^p$-graded algebra, we define the $\Z_2^p$-graded commutator 
\beq
[X, Y]_{gr}= X Y - (-1)^{(deg(X), deg(Y))}Y X
\eeq 
for any homogeneous elements $X$ and $Y$of degrees $deg(X), deg(Y)\in \Z_2^p$ respectively.  
A $\Z_2^p$-graded algebra is said to be $\Z_2^p$-graded commutative if 
$
[A,  B]_{gr}=0
$ 
for all homogeneous elements $A, B$.

We will be interested in the following family of $\Z_2^p$-graded commutative algebras. 
\begin{definition}
Given any positive integers $p$ and $n$, let $\Lambda_n(p)$ be a $\Z_2^p$-graded commutative algebra generated by the elements $\theta_i^{(\alpha)}$, for $1\le i \le n$ and $1\le \alpha\le p$, which are homogeneous of degrees $deg(\theta_i^{(\alpha)})={\bf e}_\alpha$ respectively, and satisfy the following defining relations
\beq\label{eq:comps}
\theta_i^{(\alpha)} \theta_j^{(\beta)} = (-1)^{({\bf e}_\alpha, {\bf e}_\beta)} \theta_j^{(\beta)} \theta_i^{(\alpha)}, \quad  \forall i, j, \alpha, \beta. 
\eeq
\end{definition}
Note that \eqref{eq:comps} is equivalent to
\beq
\theta_i^{(\alpha)} \theta_j^{(\alpha)} =-  \theta_j^{(\alpha)}\theta_i^{(\alpha)},   \quad \theta_i^{(\alpha)} \theta_j^{(\beta)} =  \theta_j^{(\beta)}\theta_i^{(\alpha)}  \ \text{ if $\alpha\ne\beta$}. 
\eeq
In particular, the first relation implies 
$\left(\theta_i^{(\alpha)}\right)^2=0$ for any $i$ and $\alpha$. 

For each $\alpha$, the elements $\theta_j^{(\alpha)}$ (with $j=1, 2, \dots, n$) generate a $\Z_2^p$-graded subalgebra $\Lam_n^{(\alpha)}$ of $\Lambda_n(p)$, which is isomorphic to the Grassman algebra $\Lam_n(1)$ of degree $n$ as a $\Z_2$-graded algebra.  As $\Z_2^p$-graded algebra, 
\beq\label{eq:grass-tensor}
\Lambda_n(p) = \Lambda_n^{(1)}\ot_\K \Lambda_n^{(2)}\ot_\K \dots\ot_\K \Lambda_n^{(p)}.
\eeq
Note that different $\Lambda_n^{(\alpha)}$ commute. 
If $p=1$, it is clear that $\Lambda_n(p)$ is the Grassmann algebra of degree $n$.

The following elements, for all $\epsilon=(\epsilon_1, \epsilon_2, \dots, \epsilon_n)\in \{0, 1\}^n$,  
\[
\Theta^{(\alpha), \epsilon}= \left(\theta_1^{(\alpha)}\right)^{\epsilon_1} \left(\theta_2^{(\alpha)}\right)^{\epsilon_2}\dots \left(\theta_n^{(\alpha)}\right)^{\epsilon_n}, 
\]
form a basis of $\Lam_n^{(\alpha)}$.
Write $\{0, 1\}^{n\times p}$ for the direct product $\{0, 1\}^n\times \{0, 1\}^n \times \dots \times \{0, 1\}^n$ of $p$-copies of $\{0, 1\}^n$. 
Given $\SE=(\epsilon^{(1)}, \epsilon^{(2)}, \dots, \epsilon^{(p)})$ in $\{0, 1\}^{n\times p}$, denote 
\beq\label{eq:Theta-E}
\Theta^\SE= \Theta^{(1), \epsilon^{(1)}} \Theta^{(2), \epsilon^{(2)}} \dots \Theta^{(p), \epsilon^{(p)}}. 
\eeq
Then we have the  basis 
$
\left\{\Theta^\SE \mid \SE \in \{0, 1\}^{n\times p}\right\}
$
for $\Lambda_n(p)$.

\subsection{The $\Z_2^p$-graded Clifford algebra $\W_n(p)$}
Assume that an element $\partial$ of $\End_\K(\Lambda_n(p))$ is 
homogeneous of degree $deg(\partial)\in \Z_2^p$. 
It is a $\Z_2^p$-graded derivation if for any homogeneous $A, B\in \Lambda_n(p)$, 
\beq\label{eq:grad-der}
\partial(AB)= \partial(A)B + (-1)^{(deg(\partial), deg(A))} A \partial(B). 
\eeq
Introduce the $\Z_2^p$-graded derivations $\frac{\partial}{\partial \theta_i^{(\alpha)}}$ of degrees ${\bf e}_\alpha$ such that
\beq\label{eq:grad-der-1}
\frac{\partial}{\partial \theta_i^{(\alpha)}}\left(\theta_j^{(\beta)}\right)=\delta_{i j} \delta_{\alpha \beta}, \quad  
\forall i, j, \alpha, \beta. 
\eeq
We can also regard $\Lambda_n(p)$ as a subalgebra of $\End_\K(\Lambda_n(p))$, which acts by multiplication from the left.  

\begin{definition}\label{def:Wn}
Denote by $\W_n(p)$ the $\Z_2^p$-graded subalgebra of $\End_\K(\Lambda_n(p))$ generated by $\theta_i^{(\alpha)}$ and $\frac{\partial}{\partial \theta_i^{(\alpha)}}$ for all $1\le i\le n$ and $1\le \alpha\le p$. Call $\W_n(p)$ a {\em $\Z_2^p$-graded Clifford algebra}. 
\end{definition}

One can easily prove the following result.  
\begin{lemma} \label{lem:theta-comp}  The following relations hold in $\W_n(p)$ for all $i, j, \alpha, \beta$.
\beq\label{eq:theta-comp}
\phantom{XXX} \left[\theta_i^{(\alpha)}, \theta_j^{(\beta)}\right]_{gr}=0,\quad
\left[ \frac{\partial}{\partial \theta_i^{(\alpha)}}, \frac{\partial}{\partial \theta_j^{(\beta)}}\right]_{gr}=0, \quad \left[ \frac{\partial}{\partial \theta_i^{(\alpha)}}, \theta_j^{(\beta)}\right]_{gr}= \delta_{i j} \delta_{\alpha \beta}.
\eeq
\end{lemma}
Note that the relation among $\theta_i^{(\alpha)}$'s is nothing else but \eqref{eq:comps}, which is included here as part of a presentation of $\W_n(p)$.  

Denote by ${\mathbb\Delta}_n(p)$ the subalgebra of $\W_n(p)$ generated by all $\frac{\partial}{\partial \theta_i^{(\alpha)}}$. A basis for it can be easily constructed.  
For any $\epsilon=(\epsilon_1, \epsilon_2, \dots, \epsilon_n)\in \{0, 1\}^n$,  let
\[
\partial^{(\alpha), \epsilon}= \left(\frac{\partial}{\partial\theta_1^{(\alpha)}}\right)^{\epsilon_1} \left(\frac{\partial}{\partial\theta_2^{(\alpha)}}\right)^{\epsilon_2} \dots \left(\frac{\partial}{\partial\theta_n^{(\alpha)}}\right)^{\epsilon_n}.
\]
For any $\SE=(\epsilon^{(1)}, \epsilon^{(2)}, \dots, \epsilon^{(p)})$ in $\{0, 1\}^{n\times p}$, we denote 
\beq\label{eq:partial-E}
\partial^\SE= \partial^{(1), \epsilon^{(1)}} \partial^{(2), \epsilon^{(2)}} \dots \partial^{(p), \epsilon^{(p)}}. 
\eeq
Then we have the  basis $
\left\{\partial^\SE \mid \SE \in \{0, 1\}^{n\times p}\right\}
$
for ${\mathbb\Delta}_n(p)$. 

The following elements of ${\mathbb\Delta}_n(p)$ will play a special role.
\beq\label{eq:diff-int}
D= \prod_{\alpha=1}^p D^{(\alpha)}, \quad  D^{(\alpha)}=\frac{\partial}{\partial\theta_n^{(\alpha)}} \frac{\partial}{\partial\theta_{n-1}^{(\alpha)}} \dots \frac{\partial}{\partial\theta_1^{(\alpha)}}, 
\eeq

Now $\Lam_n(p)$ and ${\mathbb\Delta}_n(p)$ are subalgebras of $\W_n(p)$.  The multiplication of $\W_n(p)$ gives a 
vector space isomorphism $\Lam_n(p)\ot_\K {\mathbb\Delta}_n(p)\lra \W_n(p)$. 
It follows that there is the following basis 
$
\left\{\Theta^\SE\partial^{\SE'} \mid \SE, \SE' \in \{0, 1\}^{n\times p}\right\}
$
for $\W_n(p)$.

\subsection{Integral on $\Lam_n(p)$}

The degree $1$ Grassmann algebra is $\Lam_1(1)=\K+\K\zeta$, where $\zeta=\zeta_1$, and hence $\frac{\partial}{\partial \zeta}$ maps $\Lam_1(1)$ to $\K$. 
 Berezin's original definition of the integral $\int_{(1)}: \Lam_1(1)\lra\K$ is to simply take it as the derivation, i.e., 
\[
\int_{(1)}  (a+ b\zeta)=\frac{\partial}{\partial \zeta}(a+b\zeta)=b, \ \text{  for all $a, b\in\K$.}
\]
This immediately generalises to arbitrary degrees, leading to the Berezin integral $\int_{(1)}: \Lam_n(1)\lra\K$ for any $n\ge 1$ defined by
\beq\label{eq:Berezin}
\int_{(1)} F(\zeta_1, \zeta_2, \dots, \zeta_n) = \frac{\partial}{\partial \zeta_n}  \frac{\partial}{\partial \zeta_{n-1}} \dots \frac{\partial}{\partial \zeta_1}F(\zeta_1, \zeta_2, \dots, \zeta_n).
\eeq

\begin{remark}
The Berezin integral has translational invariance, which however needs to be formulated over the Grassmann algebra $\Lam_\infty(1)$ of infinite degree \cite{dW}.  Extend $\frac{\partial}{\partial \zeta_i}$ to $\Z_2$-graded derivations of $\Lam_n(1)\ot_\C\Lam_\infty(1)$. Then for any odd elements $\lambda_i\in\Lam_\infty(1)$, we have
$
\frac{\partial}{\partial \zeta_i}(\zeta_j+\lambda_j)=\frac{\partial}{\partial \zeta_i}(\zeta_j)=\delta_{i j}. 
$
Now one can extend the Berezin integral \eqref{eq:Berezin} to a homogeneous map of degree $n\ ({\rm mod}\  2)$ from $\Lam_n(1)\ot_\C\Lam_\infty(1)$ to $\Lam_\infty(1)$, which satisfies 
\beq
\int_{(1)} F(\zeta_1+\lambda_1, \zeta_2+\lambda_2, \dots, \zeta_n+\lambda_n)
= \int_{(1)} F(\zeta_1, \zeta_2, \dots, \zeta_n). 
\eeq
\end{remark}

The Berezin integral can be generalised to an integral on the $\Z_2^p$-graded commutative algebra $\Lam_n(p)$.

Recall that  for each fixed $\alpha$, the elements $\theta_i^{(\alpha)}$ with $1\le i\le n$ generate a graded subalgebra $\Lam_n^{(\alpha)}(1)$ of $\Lambda_n(p)$, which is isomorphic to the Grassmann algebra $\Lam_n(1)$ as a $\Z_2$-graded algebra. Thus we have the Berezin integral 
$\int_{(1)}: \Lam_n^{(\alpha)}(1)\lra \K$. 

As $\Lambda_n(p)$ is the tensor product of all  $\Lambda_n^{(\alpha)}$ as described by \eqref{eq:grass-tensor},  we have the natural $\K$-linear map 
\beq\label{eq:int}
\int_{(p)}:=\underbrace{\int_{(1)}\ot \int_{(1)}\ot\dots\ot\int_{(1)}}_p:  \Lambda_n(p)\lra \K,  
\eeq
which will be referred to as the integral on the $\Z_2^p$-graded algebra $\Lambda_n(p)$.

We have the following result. 
\begin{lemma} \label{lem:int-der} Retain notation above. 
The integral of any $F\in \Lam_n(p)$ is given by 
\beq\label{eq:int-der}
\int_{(p)} F = D(F). 
\eeq
\begin{proof}
This easily follows from \eqref{eq:Berezin}. 
\end{proof}
\end{lemma}

Observe that $\Theta^{\SE}\in \Ker(D)$ if $\SE\ne \SE_{max}=(\underbrace{1, 1, \dots, 1}_{n p})$, and  $
D(\Theta^{\SE_{max}})=1. 
$

\subsection{Inhomogeneous subalgebras}\label{sect:inhom}
%
%
%
%

We should point out that material presented here is largely contained in \cite{OK80}, 
though not in the language of $\Z_2^p$-graded commutative algebras used here. In particular, the algebras $\ILam_n(p)$ and $\IW_n(p)$ studied below 
appeared in \cite{OK80}.

Consider the $\Z_2^p$-graded commutative algebra $\Lam_n(p)$ with generators $\theta_i^{(\alpha)}$, for $1\le i\le n$. Let 
\beq\label{eq:ansatz-theta}\label{eq:key} 
\theta_i=\sum_{\alpha=1}^p\theta_i^{(\alpha)}, \quad i=1, 2, \dots, n,  
\eeq
and denote by $\ILam_n(p)$ the subalgebra of $\Lam_n(p)$ generated by these elements. 
We will also refer to  $\ILam_n(p)$ as a  para-algebra, and call 
$\Lam_n(p)$  the {\em ambient algebra} of $\ILam_n(p)$. 
Note that $\ILam_n(p)$ is not $\Z_2^p$-graded as a subalgebra of $\Lam_n(p)$. 

\begin{lemma}\label{lem:ILam}
The algebra $\ILam_n(p)$ has the following properties:  
\beq
[\theta_k, [\theta_i, \theta_j]]=0,  \quad  \forall i, j, k,  \label{eq:the-the-2}
\eeq
and for any $\theta=\sum_i c_i \theta_i$ with $c_i\in\K$, 
\beq
\theta^p \theta_i = -\theta_i \theta^p, \quad \forall i.      \label{eq:power-p}
\eeq
\end{lemma}
\begin{proof}
Using the first relation in \eqref{eq:theta-comp}, one easily obtains
$
{[\theta_i, \theta_j]}= 2 \sum_{\alpha=1}^n \theta_i^{(\alpha)} \theta_j^{(\alpha)}$, for all $i, j.   
$
This then leads to 
$
[\theta_k, [\theta_i, \theta_j]]=0$ for all $i, j, k. 
$ 

Any $\theta=\sum_i c_i \theta_i$ can be expressed as $\theta=\sum_\alpha \theta^{(\alpha)}$ with $\theta^{(\alpha)}=\sum_i c_i \theta_i^{(\alpha)}$.  Again by using the first relation in \eqref{eq:theta-comp}, we obtain $\theta^p= p! \theta^{(1)} \theta^{(2)}\dots \theta^{(p)}$. Then it is clear that $\theta^p \theta_i=-\theta_i\theta^p$ for all $i$, completing the proof of the lemma. 
\end{proof}
Note in particular that
\beq
\theta_j^p \theta_i=-\theta_i\theta_j^p, \quad \forall i, j.
\eeq

\begin{remark}\label{rmk:vanish}
If $d_j(i_1, i_2, \dots, i_\ell)=p$ for all $j=1, 2, \dots, n$, then $\Theta_{i_1, i_2, \dots, i_\ell}$ is a scalar multiple of $\theta_1^p \theta_2^p\dots \theta_n^p= (p!)^n \Theta^{\SE_{max}}$ (see \eqref{eq:Theta-E} for notation), where $\SE_{max}=(1, 1, \dots, 1)$ is of length $np$. 
\end{remark}

\begin{remark}\label{rmk:Z-grade}
Note that $\ILam_n(p)$ admits a $\Z_+$-grading.  Write $\Lam_n(p)=\sum_{\ell\ge 0}\ILam_n(p)_\ell$ with  $deg(\theta_i)=1$ for all $i$, then $\ILam_n(p)_\ell$
is spanned by the elements $\Theta_{i_1, i_2, \dots, i_\ell}$. In particular, 
\[
\baln
&\ILam_n(p)_0=\K, \quad \ILam_n(p)_1=\sum_i \K\theta_i, \\
&\ILam_n(p)_{n p}=\K \Theta^{\SE_{max}}, \quad  \ILam_n(p)_\ell=0, \ \forall \ell>n p.  
\ealn
\]
\end{remark}

Introduce the following inhomogeneous elements in $\W_n(p)$, 
\beq\label{eq:ansatz-Delta}
\Delta_i=\sum_{\alpha=1}^p\frac{\partial}{\partial\theta_i^{(\alpha)}}, 
\quad i=1, 2, \dots, n, 
\eeq
where $\frac{\partial}{\partial\theta_i^{(\alpha)}}$ are $\Z_2^p$-graded derivations of $\Lam_n(p)$ defined by \eqref{eq:grad-der-1}. 

Denote by $\IW_n(p)$ the subalgebra of  $\W_n(p)$ generated by $\theta_i$ and $\Delta_i$ for all $i=1, 2, \dots, n$, and we will also refer to it as a para Clifford algebra.

The algebra $\IW_n(p)$ naturally acts  on $\Lambda_n(p)$. We denote the action of any $Q\in \IW_n(p)$ on any $F\in \Lam_n(p)$ by $Q(F)$.  Then
\beq\label{eq:1}
\Delta_i(1)=0, \quad 
\Delta_i(\theta_j)=p\delta_{i j}, \quad \forall i, j. 
\eeq
It also has the following properties. 
\begin{lemma}\label{lem:IW} The following relations hold in $\IW_n(p)$ for all $i, j, k=1, 2, \dots, n$.
\beq
&&   [\theta_k, [\theta_i, \theta_j]]=0,  \quad 
 [\Delta_k, [\Delta_i, \Delta_j]]=0,   \label{eq:Del-1}\\
&&[\Delta_k, [\theta_i, \Delta_j]]=2\delta_{k i}\Delta_j, 
\quad [\theta_k, [\theta_i, \Delta_j]]= -2\delta_{ k j} \theta_i. \label{eq:Del-2}
\eeq
\begin{proof} 
The first relation in \eqref{eq:Del-1} is \eqref{eq:the-the-2}, which is included here for easy reference. 
To prove the rest of the lemma, we need the following relations for all $i, j$,
\beq
&&[\Delta_i, \Delta_j]=2\sum_{\alpha=1}^p\frac{\partial}{\partial\theta_i^{(\alpha)}} \frac{\partial}{\partial\theta_j^{(\alpha)}}. \label{eq:Del-Del}, \\
&&[\Delta_i, \theta_j]=p \delta_{i j} -2  \sum_{\alpha=1}^p\theta_j^{(\alpha)} \frac{\partial}{\partial\theta_i^{(\alpha)}}. \label{eq:Del-the}
\eeq
The first relation is quite evident, but \eqref{eq:Del-the} needs some computation to verify. Using Lemma \ref{lem:theta-comp}, we obtain
\[
[\Delta_i, \theta_j]= \sum_{\alpha}\left( \frac{\partial}{\partial\theta_i^{(\alpha)}} \theta_j^{(\alpha)} - \theta_j^{(\alpha)} \frac{\partial}{\partial\theta_i^{(\alpha)}}\right) =p \delta_{i j} -2  \sum_{\alpha}\theta_j^{(\alpha)} \frac{\partial}{\partial\theta_i^{(\alpha)}}, 
\]
proving \eqref{eq:Del-the}.

Now the formulae in the lemma can be verified by straightforward calculations using Lemma \ref{lem:theta-comp} and the relations \eqref{eq:Del-Del} and \eqref{eq:Del-the}. 
Consider, for example,  the first relation of \eqref{eq:Del-2}.
Using \eqref{eq:Del-the},  we obtain 
\[
[\Delta_k, [\theta_i, \Delta_j]] = 2\sum_{\alpha, \beta} 
\left[  \frac{\partial}{\partial\theta_k^{(\beta)}}, 
\theta_i^{(\alpha)} \frac{\partial}{\partial\theta_j^{(\alpha)}}
\right].
\]
Using the standard fact that $[C, A B] = [C, A]_{gr}B + (-1)^{(deg(C), \deg(A))} A[C, B]_{gr}$ for all homogeneous elements $A, B, C$, we can re-write the right hand side as 
\[
\baln
2\sum_{\alpha, \beta} 
\left[  \frac{\partial}{\partial\theta_k^{(\beta)}}, 
\theta_i^{(\alpha)}  \right]_{gr} \frac{\partial}{\partial\theta_j^{(\alpha)}}
   + 2\sum_{\alpha, \beta} 
(-1)^{({\bf e}_\alpha, {\bf e}_\beta)} \theta_i^{(\alpha)} \left[  \frac{\partial}{\partial\theta_k^{(\beta)}}, 
 \frac{\partial}{\partial\theta_j^{(\alpha)}}
\right]_{gr}.
\ealn
\]
Now use Lemma \ref{lem:theta-comp} in the expression. The second term vanishes, and the first term is equal to
$
2\sum_{\alpha, \beta} \delta_{k i} \delta_{\alpha \beta}
\frac{\partial}{\partial\theta_j^{(\alpha)}} = 2 \delta_{k i}\Delta_j, 
$
proving the first relation of  \eqref{eq:Del-2}. 
\end{proof}
\end{lemma}

We note that \eqref{eq:Del-2} implies
\beq
&&   [\Delta_k, [\theta_i, \theta_j]]=2\delta_{k i}\theta_j-2\delta_{k j} \theta_i, \label{eq:Del-3}\\
&& [\theta_k, [\Delta_i, \Delta_j]] = 2\delta_{k i} \Delta_j - 2\delta_{k j} \Delta_i.  \label{eq:Del-4}
\eeq
By using the second relation in \eqref{eq:theta-comp}, and similar arguments as those in the proof of Lemma \ref{lem:Xi-ILam}, we can show that  for any $\Delta=\sum_i c_i\Delta_i$ with $c_i\in\K$, 
 \beq \label{eq:Del-p-1}
\Delta^p \Delta_j= -  \Delta_j \Delta^p, \quad \forall  j. 
\eeq


Now we can extract an integral on $\ILam_n(p)$ from the integral $\int_{(p)}$ on the ambient algebra $\Lam_n(p)$.

\begin{definition}\label{def:int-ILam} Denote the restriction of the $\K$-linear map \eqref{eq:int} to $\ILam_n(p)$ by 
\[
\int_{(p)}: \ILam_n(p)\lra \K,
\]
 and call it an integral on the para-algebra $\ILam_n(p)$ of order $p$. 
\end{definition}

The following result immediately follows from Lemma \ref{lem:int-der} and \eqref{eq:diff-int-1}.  
\begin{theorem} \label{thm:int-p} For any $F\in \ILam_n(p)$, 
\[
\int_{(p)} F = \frac{1}{(p!)^n}\Delta_n^p \Delta_{n-1}^p \dots \Delta_1^p(F). 
\]
Hence Definition \ref{def:int-ILam}  is independent of the ambient algebra $\Lam_n(p)$.
\end{theorem}

We can easily evaluate the integral on the spanning elements $\Theta_{i_1, i_2, \dots, i_\ell}$ of $\ILam_n(p)$. 
\begin{theorem} Retain notation in Theorem \ref{thm:vanish}. Then
\[
\int_{(p)}\Theta_{i_1, i_2, \dots, i_\ell}\ne 0 \ \text{  only if $d_j(i_1, i_2, \dots, i_\ell)=p$ for all $j$}, 
\]
and  in particular, 
$\int_{(p)}
\theta_1^p \theta_2^p\dots \theta_n^p=(p!)^n$.  
\end{theorem}
\begin{proof}
We express $\Theta_{i_1, i_2, \dots, i_\ell}$ as a linear combination of basis elements of $\Lam_n(p)$ given by \eqref{eq:Theta-E}. Then $\Theta^{\SE_{max}}$ appears with non-zero coefficient only if $d_j(i_1, i_2, \dots, i_\ell)=p$ for all $j$. In particular,  $\theta_1^p \theta_2^p\dots \theta_n^p=(p!)^n\Theta^{\SE_{max}}$ and $\int_{(p)}
\theta_1^p \theta_2^p\dots \theta_n^p=(p!)^n$. 
\end{proof}

\begin{remark}
There is a Hopf algebraic construction \cite[Example 1]{SchZ01} of the Berezin integral on any Grassmann algebra. The same construction works for classical and quantum supergroups \cite{SchZ01, SchZ05}, and can also be generalised to para-algebras of any finite degrees and orders. 
\end{remark}

\medskip

\subsection*{Acknowledgements} 
We thank Naruhiko Aizawa, Alan Carey, Mark Gould, Peter Jarvis and Joris Van der Jeugt for comments. 

\medskip

\subsection*{Declaration}
The authors declare that they have no known competing financial interests or personal relationships 
that could have appeared to influence the work reported in this paper.

\medskip




\end{document}